\documentclass[11pt,a4paper,final]{article}

\usepackage{amsmath}    
\usepackage{amsfonts}   
\usepackage{amssymb}    
\usepackage{amsthm}     
\usepackage{mathrsfs}   
\usepackage{parskip}    
\usepackage{calc}       
\usepackage{enumitem}   
\usepackage{graphicx}   
\usepackage{lpic}       
\usepackage{caption}
\usepackage{color}      
\usepackage{mathtools}
\usepackage{thmtools,thm-restate} 
\usepackage{subcaption}
\usepackage{rotating}
\usepackage{placeins}
\usepackage[hidelinks]{hyperref}
\usepackage[authoryear,round,longnamesfirst]{natbib}

\let\Oldsection\section
\renewcommand{\section}{\FloatBarrier\Oldsection}
\let\Oldsubsection\subsection
\renewcommand{\subsection}{\FloatBarrier\Oldsubsection}
\let\Oldsubsubsection\subsubsection
\renewcommand{\subsubsection}{\FloatBarrier\Oldsubsubsection}

\newcommand{\ie}{{\it i.e.}}
\newcommand{\cf}{{\it cf.}}
\newcommand{\eg}{{\it e.g.}}

\newcommand{\iid}{{\it i.i.d.}}

\newcommand{\PP}{{\mathbb P}}

\newcommand{\RR}{{\mathbb R}}

\newcommand{\NN}{{\mathbb N}}

\newcommand{\FB}{\mathbb{F}_B}
\newcommand{\FA}{\mathbb{F}_A}

\newcommand{\DB}{\mathbb{D}_{B}}
\newcommand{\DA}{\mathbb{D}_{A}}

\newcommand{\conv}[1]%
  {{\mathrel{\,\xrightarrow{\widthof{\,#1\,}}\,}}}
\newcommand{\convas}[1]%
  {{\mathrel{\,\xrightarrow{\widthof{\,#1\text{-a.s.}\,}}\,}}}
\newcommand{\convprob}[1]%
  {{\mathrel{\,\xrightarrow{\widthof{\,#1\,}}\,}}}
\newcommand{\convweak}[1]%
  {{\mathrel{\,\xrightarrow{\widthof{\,#1\text{-w.}\,}}\,}}}

\renewcommand{\qedsymbol}{$\Box$}

\newtheoremstyle{customtheorem}
  {0.5em}
  {0.2em}
  {\itshape}
  {}
  {\scshape}
  {}
  {1ex}
  {}

\theoremstyle{customtheorem}
\newtheorem{theorem}{Theorem}[section]

\newtheorem{proposition}[theorem]{Proposition}

\newtheorem{definition}[theorem]{Definition}
\newtheoremstyle{customremark}
  {0.5em}
  {0.2em}
  {}
  {}
  {\scshape}
  {}
  {1ex}
  {}

\theoremstyle{customremark}
\renewenvironment{proof}{\par\noindent{\scshape Proof}\;}{\hfill\qedsymbol\par}
\newtheorem{remark}[theorem]{Remark}

\newtheorem{assumption}{Assumption}

\begin{document}
\thispagestyle{empty}

\title{Heavy tailed distributions in closing auctions}
\author{
  M. Derksen${}^{1,2}$, B. Kleijn${}^{2}$ and R. de Vilder${}^{1,2}$\\[1mm]
  {\small\it  ${}^{1}$ Deep Blue Capital N.V., Amsterdam}\\
  {\small\it  ${}^{2}$ Korteweg-de~Vries Institute for Mathematics,
    University of Amsterdam}
  }
\date{\today}
\maketitle

\begin{abstract} \noindent
We study the tails of closing auction return distributions for a sample of liquid European stocks. We use the stochastic call auction model of \cite{Derksenetal20a}, to derive a relation between tail exponents of limit order placement distributions and tail exponents of the resulting closing auction return distribution and we verify this relation empirically. Counter-intuitively, large closing price fluctuations are typically not caused by large market orders, instead tails become heavier when market orders are removed. 
The model explains this by the observation that limit orders are submitted so as to counter existing market order imbalance.
\end{abstract}

\section{Introduction}
 During the trading day, most securities change hands in continuous double auctions, in which buy and sell orders are immediately matched if possible. However, to determine opening and closing prices, call auctions are often conducted. In a call auction, orders are aggregated for an interval of time, after which all possible transactions are conducted against a single clearing price that maximizes trading volume. In this paper we study the tails of closing auction return distributions.
 
Nowadays it is widely recognized that distributions of (stock) price changes exhibit heavy tails: extreme price changes (of \eg\ more than three standard deviations) are much more likely than in a Gaussian model or other models with exponentially decaying tails. This issue was first adressed by \citet{Mandelbrot63} in his analysis of cotton prices, where he proposed L\'evy stable distributions to model price fluctuations. It is generally assumed that the tails follow a power law asymptotically. That is, the distribution of a return $X$ over some time interval satisfies\footnote{Here, $\sim$ denotes \emph{asymptotic equivalence}, defined as $f \sim g \Leftrightarrow \lim_{x \to \infty} \frac{f(x)}{g(x)}=1$.},
\begin{align} \label{eq:powerlaw}
\PP(X>x) \sim Cx^{-a}, \text{ as } x \to \infty,
\end{align}
 where $C>0$ is a constant (sometimes also replaced by a slowly varying factor $L(x)$) and $a>0$ is the \emph{tail exponent}, determining how heavy the tail is . In early work \citep{Fama65}, the exponent $a$ was believed to be below 2 for stock prices (in line with the stable distributions of \cite{Mandelbrot63}). However, subsequent analyses have shown that the exponent is more likely to be around 3 on intraday  time scales (see \eg\ \cite{Gopikrishnanetal98,Gopikrishnanetal99,Guetal08,Pagan96,PlerouStanley08}, among many others). Although it is generally accepted to model the tails as power laws, the exact functional form is also subject of debate. For example, \cite{Malevergneetal05} conclude that the tails decay slower than stretched exponential distributions, but somewhat faster than power laws. In this paper, we do not aim to answer this question, but use power laws because they describe the tails in enough detail for our analysis. Theoretically, the functional form in equation (\ref{eq:powerlaw}) is justified by extreme value theory, in the Fr\'echet (heavy tailed) case (see \eg\ \cite{Embrechtsetal03}). 
 
 Although most part of the relevant literature focuses on description of the tails of stock price return distributions, some effort has gone towards explanations of this tail behaviour. \cite{Gabaixetal03,Gabaixetal06} argue that large price fluctuations are due to large orders submitted by large market participants. However, \cite{Farmeretal04} and \cite{WeberRosenow06} study the issue on the microscopic level and find that large returns are not due to large transactions, but instead are caused by big gaps in the order book, \ie\ fluctuations in liquidity.  \cite{MikeFarmer08} propose a simulation based model for continuous trading, which suggests heavy tails in return distributions are caused by market microstructure effects, such as heavy tails in limit order placement and long memory in order flow. More theoretically, \cite{Baketal97} and \cite{ContBouchaud00} propose models linking heavy tails to herd behaviour.

\subsection{Main results}
In this paper, we use the model of \cite{Derksenetal20a} to study the distribution of returns in the closing auction. In the model, limit orders are submitted to the auction randomly, with a limit price that is sampled from an \emph{order placement distribution} $F_A$ (for sell orders) or $F_B$ (for  buy orders). 
We study the closing auctions of liquid European stocks listed on Euronext exchanges and find that both return distributions and order placement distributions exhibit heavy tails, with different tail exponents. \cite{ZovkoFarmer02} conclude \emph{`It seems that the power law for price
fluctuations should be related to that of relative limit prices,
but the precise nature and the cause of this relationship is not
clear.'} Here, we solve this problem in the context of the closing auction: we provide analytical relations between the tails of order placement distributions and the tails of the closing price return distribution. In a version of the model without market orders, the tails of the closing price distribution behave as the product of the tails of the order placement distributions $F_A$ and $F_B$. When we incorporate market orders, this relation changes, depending on a proportionality relation between market order and limit order imbalances. We empirically verify the relations between tail exponents of order placement and auction return distributions predicted by the model.

 In theory, large market orders are a possible cause of large price fluctuations. We show however that this is typically not the case in closing auctions, which is our second important result. Somewhat counter-intuitively, the empirical study shows that closing auction return distributions would have \emph{heavier} tails if market orders are removed, suggesting that market orders have a stabilizing effect on price formation in closing auctions.
Theoretically, we show (for the right tail) that this (initially perhaps somewhat puzzling) empirical fact can only arise whenever
\begin{align}\label{eq:heavier_condition0}
0 <  \frac{M_B-M_A}{N_A-N_B} \leq \frac{a_A}{a_B},
\end{align} 
under the assumption that $F_B$ and $F_A$ have heavy right tails with tail exponents $a_B$ and $a_A$ satisfying $a_B>a_A>0$.
Here, $N_A$ is the sell limit order volume, $N_B$ the buy limit order volume and $M_A$ and $M_B$ denote the sell and buy market order volume. This equation poses two conditions that should be fulfilled to make it theoretically possible that tails of closing auction return distributions are heavier without market orders. First, limit order imbalance and market order imbalance should be of opposite signs (when $M_B>M_A$, it should hold that $N_A>N_B$ and vice versa) and limit order imbalance $N_A-N_B$ should be larger in absolute value than market order imbalance $M_B-M_A$, meaning that limit orders overcompensate for market order imbalance. Second, $a_B$ should not be too large, \ie\ the right tail of the buy limit order placement distribution needs to be sufficiently heavy. We show that equation (\ref{eq:heavier_condition0}) is indeed satisfied on average empirically, which is explained by the chronology of the closing auction: most of the market orders are submitted in the first seconds, revealing early in the auction the market order imbalance. This leads to strategic behaviour in which limit orders are placed \emph{against} the direction of the market order imbalance: when there are more buy than sell market orders, one can submit a (possibly large) sell order without adversely impacting the price.  
Our results suggest that large closing price fluctuations are not caused by large market orders (at least, not directly), but by placement of limit orders, in accordance with the intraday results of \cite{Farmeretal04} and \cite{WeberRosenow06}. Also, our results suggest that heavy tails are market microstructure effects and that the tail exponents vary between different stocks and different market mechanisms, in line with the view of \cite{MikeFarmer08}.

The remainder of this paper is structured as follows. In section \ref{sec:theory} the model is described and theoretical results are derived. Then in section \ref{sec:empirics} the empirical results are presented and the relations that are predicted by the model are verified. Concluding remarks are made in section \ref{sec:conclusion} and proofs of the mathematical theory are collected in the appendix.
\section{Theoretical results}\label{sec:theory}
In this section we recall the auction model of \cite{Derksenetal20a} and derive analytical expressions for the tail behaviour of the return distribution, given the tails of order placement distributions.
\subsection{A stochastic model of the call auction}
 In the standard call auction, orders are aggregated over an interval of time and then matched to transact at
a clearing price that maximizes the total transacted volume. Suppose $N_A$ sell limit orders and $N_B$ buy limit orders are submitted to the auction (all orders have unit size). We assume market participants on both sides of the market formulate their orders independently, according to certain order placement distributions $F_A$ and $F_B$. Here, $F_A$ denotes the distribution of sell orders and $F_B$ the distribution of buy orders. That is, we model the sell order prices $(A_1,\dots,A_{N_A})$ as an \iid\ sample from $F_A$ and the buy order prices $(B_1,\dots,B_{N_B})$ as an \iid\ sample from $F_B$\footnote{Of course these assumptions are not all realistic. In reality, orders have different sizes and market participants may react to each other's orders. Despite these simplifying assumptions, the model provides a reliable stochastic description of auction price formation (see \cite{Derksenetal20a}).}.

For convenience we consider the log return axis instead of the real price axis. We assume there is some \emph{reference price} $x_0$ (for example the last traded price before the auction starts or a volume weighted averaged version thereof) and all prices are expressed as log returns relative to this reference price. So $F_A$ and $F_B$ are distributions on $(-\infty,\infty)$ and $F_A(x)$ or $F_B(x)$ denotes the probability that a sell or buy order price is below $x_0e^x$.
Given $(N_A,N_B)$, we denote by $\FA$ and $\FB$ the empirical distribution functions corresponding to the samples $(A_1,\dots A_{N_A})$ and $(B_1,\dots, B_{N_B})$, meaning
\begin{align*} \FA(x) = \frac{1}{N_A} \sum_{i=1}^{N_A} \mathbf{1}_{\{A_i \leq x\}}, &\quad  \FB(x) = \frac{1}{N_B} \sum_{i=1}^{N_B} \mathbf{1}_{\{B_i \leq x\}}
\end{align*}
Furthermore, we define the (monotone increasing) \emph{supply curve},
$$\DA(x) = N_A\FA(x)$$
and the (monotone decreasing) \emph{demand curve},
$$  \DB(x) = N_B(1-\FB(x)).$$
The supply curve denotes for every $x \in \mathbb{R}$ the number of sell orders below $x_0e^x$, the demand curve gives for every $x \in \mathbb{R}$ the number of buy orders above $x_0e^x$. Given all buy and sell orders, the clearing price is the price that maximizes the transactable volume in the auction, which is the price where supply and demand curves cross. That is, the \emph{clearing price} $X$ is defined as the solution to the market clearing equation,
\begin{align}\label{eq:clearingeq} \DA(X) = \DB(X).
\end{align}
This definition of $X$ may give rise to problems with uniqueness and existence of solutions to equation (\ref{eq:clearingeq}), as illustrated in figure \ref{fig:da_db_simple}. To solve these issues, consider the following definition.
\begin{definition}
For given supply curve $\DA$ and demand curve $\DB$, the \emph{lower clearing price} is defined by
\begin{align} 
\underline{X}& = \inf\{x \in \RR: \DA(x) \geq \DB(x))\}
\end{align}
and the \emph{upper clearing price} is defined by
\begin{align}
\overline{X}& = \sup\{x \in \RR: \DA(x) \leq \DB(x) \} \nonumber \\
& =  \inf\{x \in \RR: \DA(x) > \DB(x))\}.
\end{align}
The interval $[\underline X,\overline{X})$ is the interval of all possible clearing prices.
\end{definition}
\begin{figure}[htb]
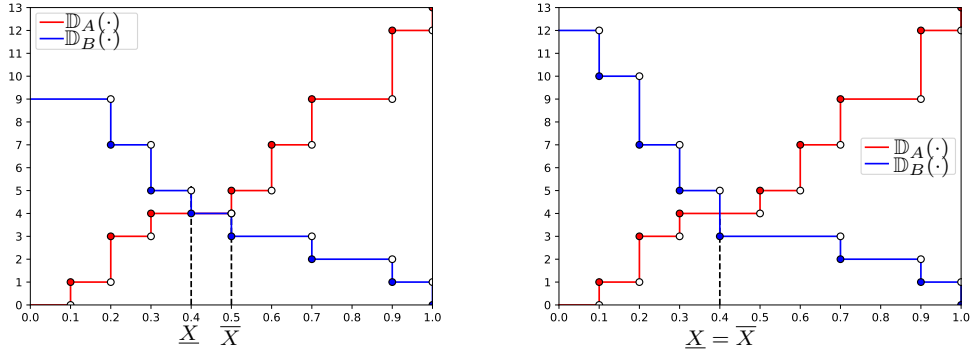

	\centering
    \begin{lpic}{da_db_simple2(0.42)}
      \lbl[t]{40,106;{\scriptsize{$\DA(\cdot)$}}}
      \lbl[t]{40,100.5;{\scriptsize{$\DB(\cdot)$}}}      
      \lbl[t]{70,7;{\scriptsize{$\underline X$}}}
      \lbl[t]{83,7;{\scriptsize{$\overline X$}}}
    \end{lpic}
    \begin{lpic}{da_db_cross2(0.42)}
      \lbl[t]{134,66;{\scriptsize{$\DA(\cdot)$}}}
      \lbl[t]{134,60.5;{\scriptsize{$\DB(\cdot)$}}}   
      \lbl[t]{71,7;{\scriptsize{$\underline X=\overline X$}}}
    \end{lpic}
    \caption[Stylized examples of a call auction.]{
      \label{fig:da_db_simple}
      Two examples of the supply curve $\DA(\cdot)$ (the increasing
      (red) step function) and the demand curve $\DB(\cdot)$ (the decreasing
      (blue) step function).
      Left panel: a situation in which there is no unique point of
      intersection, but an interval $[\underline X,\overline X)$ of possible clearing prices. Right panel: a situation in which there is a unique
      intersection point $\underline{X}=\overline X$.}
\end{figure}
\begin{remark} \label{remark:clearingprices}
 Euronext's closing auction rules say that when there are more possible clearing prices, the price closest to the last traded price is taken \citep[Rule 4401/3]{EN19}. This means that when there is a large positive return, the closing price is equal to the lower clearing price $\underline X$. So in order to study the right tail of the closing price return distribution, we should study $\underline X$. The same reasoning implies that for the left tail we should consider $\overline X$. Note that the model is symmetric when the roles of $\overline X$ and $\underline X$ and the sides of the market are interchanged. That is, the left tail of the distribution of $\overline X$ behaves the same as the right tail of the distribution of $\underline X$, when  $F_A$ and $F_B$ and $N_A$ and $N_B$ are interchanged. So without loss of generality we focus on the right tail of $\underline X$.
\end{remark}
The distribution of the lower clearing price, conditional on $(N_A,N_B)$, has an analytically tractable distribution function, given in the following theorem (see \cite{Derksenetal20a}, theorem 2.3).
\begin{restatable}[Lower clearing price distribution]
{theorem}{thmclearingpricedists} \label{thm:clearingpricedists}
The distribution of the lower clearing price $\underline{X}$, conditional on $(N_A,N_B)$, is given by its survival function,
\begin{align} \label{eq:clearingpricedist}
\PP&(\underline{X}>M|N_A,N_B)\nonumber \\ 
&  =\sum_{k=0}^{N_A} \sum_{l=k+1}^{N_B}{N_A \choose k} {N_B \choose l} (1-F_A(M))^{N_A -k} F_A(M)^k (1-F_B(M))^l F_B(M)^{N_B-l}.
\end{align}
\end{restatable}
In the situation described above, only limit orders are submitted to the auction. However, market participants also have the possibility to submit market orders.
We define the (possibly stochastic) \emph{market order imbalance} by $\Delta= M_B-M_A$, where $M_B$ is the number of buy market orders and $M_A$ is the number of sell market orders. Note that market orders only play a role through $\Delta$, as matching market orders are executed against each other without affecting the price formation process. When market order imbalance $\Delta$ is taken into account, the market clearing equation (\ref{eq:clearingeq}) becomes
  $$\DA(X) = \DB(X) + \Delta$$
 and the definitions of $\underline X$ and $\overline X$ change accordingly.
A positive (negative) value of $\Delta$ means there is more buy (sell) market order volume than sell (buy) market order volume, possibly pushing the price up (down).
The market order imbalance alters the clearing price distribution as in the following proposition (a special case of proposition 2.8 in \cite{Derksenetal20a}).
\begin{proposition} [Lower clearing price distribution in case of market order imbalance] 
\label{prop:eqprice_el_dist}
When market order imbalance $\Delta$ plays a role, the lower clearing price distribution as computed in theorem \ref{thm:clearingpricedists} modifies into
\begin{align*}
\PP&(\underline{X}>M|N_A,N_B,\Delta)\\
&= \sum_{k=0}^{N_A} \sum_{l=\max(k-\Delta+1,0)}^{N_B}{N_A \choose k} {N_B \choose l} (1-F_A(M))^{N_A -k} F_A(M)^k (1-F_B(M))^l F_B(M)^{N_B-l}.
\end{align*}
\end{proposition}
\subsection{Limit order auctions} \label{subsec:only_limit_orders}
Next we concentrate on the right tail of the lower clearing price return distribution, as a function of the tails of the order placement distributions $F_A$ and $F_B$, initially without market orders. We make the following assumption on the tails of $F_A$ and $F_B$.
\begin{assumption} \label{ass:fa_heavier}
Assume $F_A$ has a heavier right tail than $F_B$. That is, there exists functions $T_A,T_B$ such that
\begin{align*} 1-F_A(M) \sim T_A(M),\quad 1-F_B(M) \sim T_B(M),\quad  \text{ as } M \to \infty
\end{align*}
and
\begin{align*}
\lim_{M \to \infty} \frac{T_B(M)}{T_A(M)}=0.
\end{align*}
\end{assumption}
This assumption is intuitively reasonable and empirically verified in section \ref{sec:empirics_fafb}.
Furthermore, we will assume that $(N_A,N_B)$ follows a distribution $P_{N_A,N_B}$ on $$\mathcal N= \{1,\dots,N\} \times \{1,\dots,N\},$$ for some $N \in \NN$, with probability mass function $p_{N_A,N_B}$ assigning positive probability to any point in $\mathcal N$ (we exclude the possibilities that $N_A=0$ or $N_B=0$, which describe failing auctions in which clearing prices do not exist).

In the following proposition we first derive an expression for the right tail of the lower clearing price distribution, conditional on $(N_A,N_B)$. Finding an expression for the tail of the clearing price distribution amounts to finding the slowest decaying term in the double sum of theorem \ref{thm:clearingpricedists}. This is made formal in the following proposition, the proof of which is found in the appendix.
\begin{restatable}{proposition}{propcondtails} \label{prop:fa_heavier}
Under assumption \ref{ass:fa_heavier}, we have
\begin{align}
\PP(\underline{X}>M|N_A,N_B) \sim N_B T_B(M) T_A(M)^{N_A}, \text{ as } M \to \infty.
\end{align}
\end{restatable}
\begin{remark}
The proof of proposition \ref{prop:fa_heavier} reveals the event that corresponds to the slowest decaying term in the double sum of theorem \ref{thm:clearingpricedists}, namely $l=1,k=0$, corresponding to the event that $\DA(M) = 0,~\DB(M)=1$, meaning all sell orders, but only one buy order, are above $M$. This is interpreted as an auction in which there is little consensus between both sides of the market (buy and sell orders do not overlap), but there is a very aggressive buyer willing to pay a high price.
\end{remark}
When the conditional result of proposition \ref{prop:fa_heavier} is summed with respect to the distribution of $(N_A,N_B)$,  the unconditional tail of $\underline X$ is discovered again by selecting the slowest decaying term. This leads to the main result of this subsection, a relation between the tail of the closing price return distribution and the tail of the order placement distributions in a setting without market orders (its proof is again postponed to the appendix). 
\begin{restatable}[Right tail of the lower clearing price distribution]{theorem}{thmtails}\label{thm:tail_x_heavier_fa}
Under assumption \ref{ass:fa_heavier} we have
$$\PP(\underline{X}>M) \sim C T_A(M) T_B(M), \text{ as } M \to \infty,$$
where $C=\sum_{n=1}^Nn p_{N_A,N_B}(1,n) = \mathbb E[N_B \mathbf{1}_{\{N_A=1\}}]>0$.
\end{restatable}
The constant $C$ indicates that the slowest decaying term in the sum corresponds to the event that $N_A=1$: large positive returns are possible if there are only few sell orders.
\subsection{Market orders} 
In this subsection we incorporate market orders in the derivation of subsection \ref{subsec:only_limit_orders}. First consider the following assumption for the market order imbalance $\Delta$.
\begin{assumption} \label{ass:delta_bounds}
We assume that $\Delta \in (-N_B,N_A)$ with probability one.
\end{assumption}
This assumption is necessary, because otherwise the clearing prices attain the values $\pm \infty$ with non-zero probability.
Under this assumption, the right tail of the conditional lower clearing price distribution is given by the next proposition (the proof is again postponed to the appendix and $x_+ = \max(x,0)$ and $x_- = \max(-x,0)$ denote the positive and negative part of $x \in \RR$).
\begin{restatable}{proposition}{proptailsmos} \label{prop:fa_heavier_plus_mo}
Under assumptions \ref{ass:fa_heavier} and \ref{ass:delta_bounds}, we have
\begin{align}
\PP(\underline{X}>M|N_A,N_B,\Delta) \sim K(N_A,N_B,\Delta-1) T_B(M)^{(\Delta-1)_-} T_A(M)^{N_A-(\Delta-1)_+},
\end{align}
as $M \to \infty$, where $$K(N_A,N_B,\Delta) = \begin{cases} {N_A \choose \Delta} & \text{ if } \Delta>0 \\
{N_B \choose -\Delta} & \text{ if } \Delta\leq0
\end{cases}.$$
\end{restatable}
This proposition shows that market orders potentially influence the tails heavily: if $\Delta$ is positive and large (close to $N_A$) the influence of the faster decaying term $T_B(M)$ is erased and only the slower decaying term $T_A(M)$ is left, possibly leading to very heavy tails. On the other hand, if $\Delta$ is negative, the influence of the faster decaying term $T_B$ grows, leading to less heavy tails. However, which combinations are possible depends on the joint distribution of $(N_A,N_B,\Delta)$.
Until now, the tails $T_A$ and $T_B$ were unspecified and few assumptions were made on the distribution of $(N_A,N_B)$. To work towards an empirically testable theory, we will make the following assumptions on the distribution of $(N_A,N_B,\Delta)$ and the tails of $F_A,F_B$. Empirically, these assumptions are verified in section \ref{sec:empirics}.
\begin{assumption} \label{ass:delta_general}
Assume $(N_A,N_B,\Delta)$ follows a distribution $P$ on $\{1,\dots,N\} \times \{1,\dots,N\} \times \{-N,\dots,N\}$, with probability mass function denoted by $p$, for some $N \in \NN$. Furthermore, assume that market order imbalance $M_B-M_A$ is proportional to limit order imbalance $N_A-N_B$ (in the opposed direction), that is,
\begin{align} \label{eq:delta} \Delta=M_B-M_A=c (N_A-N_B),
\end{align}
almost surely for some $c \in (0,1)$  and $P(\Delta =0)=0$ (as the case $\Delta=0$ is already considered in subsection \ref{subsec:only_limit_orders}).
  Finally, assume that all possible combinations have positive probability, \ie\
$$p(n,m,d)>0,~\text{for all } n,m \in \{1,\dots,N\}, d \in \pm \{1,\dots,N\} \text{ such that } d=c(n-m).$$
\end{assumption}
Equation (\ref{eq:delta}) states that limit order imbalance points in the opposed direction of market order imbalance, which resembles that limit order submitters adjust their orders to the market order imbalance. This relation ensures assumption \ref{ass:delta_bounds} holds and is empirically verified in section \ref{subsec:cvalues}.
\begin{assumption} \label{ass:fa_heavier_pl}
Assume $F_A,F_B$ both have power law tails, that is, 
\begin{align*}
1-F_A(M) \sim T_A = M^{-a_A}, 1-F_B(M) \sim T_B(M)=M^{-a_B}, \text{ as } M \to \infty,
\end{align*}
for tail exponents $a_B>a_A>0$.
\end{assumption}
Under assumptions \ref{ass:delta_general} and \ref{ass:fa_heavier_pl}, the following theorem (which is proved in the appendix) describes the tail behaviour of the clearing price distribution in terms of the parameters $c$ (controlling the relation between market and limit order imbalance) and $a_A$ and $a_B$ (controlling the heaviness of the tails of the buy and sell limit order placement distribution).

\begin{restatable}[Right tail of the lower clearing price distribution with market orders]{theorem}{thmtailsmos} \label{thm:tailswithmos}
Under assumptions \ref{ass:delta_general} and \ref{ass:fa_heavier_pl}, there exists a constant $\underline{C}>0$, such that
$$\PP(\underline{X}>M) \sim \underline{C} M^{-\underline{a}}, \text{ as } M \to \infty,$$
where 
\begin{align} \underline{a} = \min\left(\frac{(c+1)a_A}{c}, a_A + 2a_B\right).
\end{align}
\end{restatable}
Note that without market order imbalance $\Delta$ we have by theorem \ref{thm:tail_x_heavier_fa} $\underline{a} = a_A+a_B$.  This theorem makes testable predictions about the relation between the tails of the closing price return distribution, the tails of the limit order placement distributions and the limit and market order imbalance. In the next section we will investigate this relation empirically.

\section{Empirical results} \label{sec:empirics}
In this section we investigate empirically the relation between the tails of the closing auction return distributions and the tails of the limit order placement distributions. In order to do so, we obtain detailed order-by-order data over 2018 and 2019, for 100 liquid European stocks (with market capitalization above EUR 1 bn) listed on Euronext exchanges in Amsterdam, Paris, Brussels or Lisbon.

Estimating the tails of a distribution comes with a couple of problems. First, the power law of equation (\ref{eq:powerlaw}) is not assumed to hold for all values of $x$, but only for the tail. This necessarily involves a starting point $x_{\text{min}}$ such that the power law holds for all $x>x_{\text{min}}$ (see \cite{Newman05} for a discussion). Unfortunately, the eventual estimate for the tail exponent will depend on this cut-off point: if $x_{\min}$ is taken too small, the bulk instead of the tail will determine the estimates. The cut-off  is often made through visual inspection of a double logarithmic plot. Then the second problem arises, because the cut-off eliminates most of the available data, leaving only a small fraction of the data available for estimation. Finally, models are often designed to describe only `generic' situations well and are not intended to explain extreme events. It is a noteworthy advantage of the call auction model of section \ref{sec:theory} that it is suitable to model both the bulk of the data (as in \cite{Derksenetal20a}) and extreme events, as in the current paper. 

Concerning the amount of data relevant for the tails, in every closing auction a large amount of orders is submitted, so the tails of order placement distributions can be studied per stock. Unfortunately, this is not possible for the closing auction return distribution: per stock, we have only around 500 trading days (two years of around 250 trading days per stock) and thus only that many closing auction returns, which is far insufficient to examine the tails. For example, if we take the 0.05-quantile for the cut-off point $x_{\min}$, only about 25 data points reside in the tail, which is too few for meaningful statistical analysis. So to investigate the tails of the closing auction return distribution, we merge together the closing auction returns of all stocks in the sample.

In the entire section, the reference price $x_0$ will be the \emph{volume weighted average price} over the last five minutes of  continuous trading.  Closing auction returns will be measured in log returns with respect to $x_0$. Following \cite{Bouchaudetal02} and \cite{ZovkoFarmer02}, limit order prices are measured in the number of ticks a limit order is placed away from the reference price $x_0$.
\subsection{Tails of order placement distributions}\label{sec:empirics_fafb}
The mechanism of the call auction makes it possible to study \emph{both} tails of \emph{both} order placement distributions. In figure \ref{fig:tails_fafb_4s}, both  tails of the sell limit order distribution $F_A$ and the buy limit order distribution $F_B$ are shown in log-log plots, for four stocks that are representative for the sample.

\begin{figure}[htp]
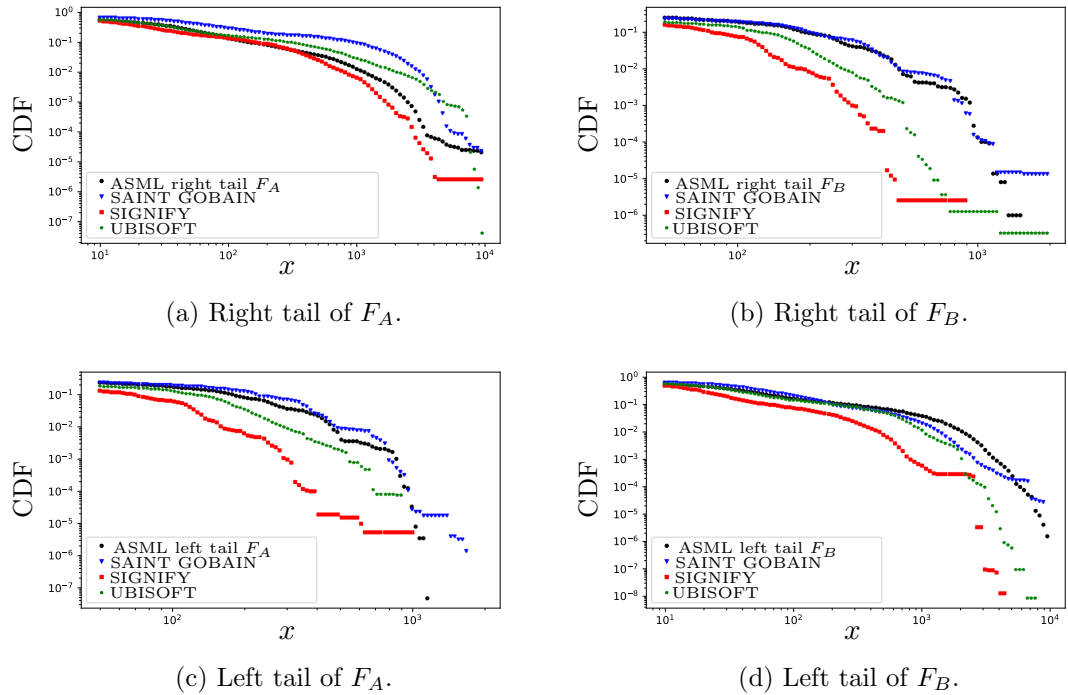

  		\begin{subfigure}[h]{0.5\textwidth}
  		\centering
      \begin{lpic}{fa_right_4s(0.4)}    
       \lbl[t]{{90,5};$x$}
       \lbl[t]{{3,60};\begin{turn}{90}\small{CDF}\end{turn}}
      \lbl[t]{{60,33};\tiny{ASML right tail $F_A$}}   
      \lbl[t]{{56,28};\tiny{SAINT GOBAIN}}
       \lbl[t]{{45,23};\tiny{SIGNIFY}}   
          \lbl[t]{{46,18};\tiny{UBISOFT}}   
      \end{lpic}
      \caption{Right tail of $F_A$.}
      \end{subfigure}
      \begin{subfigure}[h]{0.5\textwidth}
      \centering
      \begin{lpic}{fb_right_4s(0.4)}    
       \lbl[t]{{90,5};$x$}
       \lbl[t]{{3,60};\begin{turn}{90}\small{CDF}\end{turn}}
      \lbl[t]{{60,33};\tiny{ASML right tail $F_B$}}   
      \lbl[t]{{56,28};\tiny{SAINT GOBAIN}}   
       \lbl[t]{{45,23};\tiny{SIGNIFY}}   
          \lbl[t]{{45,18};\tiny{UBISOFT}}    
      \end{lpic}
        \caption{Right tail of $F_B$.}
      \end{subfigure}
      \begin{subfigure}[h]{0.5\textwidth}
      \centering
            \begin{lpic}{fa_left_4s(0.4)}    
       \lbl[t]{{90,5};$x$}
       \lbl[t]{{3,60};\begin{turn}{90}\small{CDF}\end{turn}}
      \lbl[t]{{59,33};\tiny{ASML left tail $F_A$}}   
      \lbl[t]{{56,28};\tiny{SAINT GOBAIN}}
       \lbl[t]{{45,23};\tiny{SIGNIFY}}   
          \lbl[t]{{46,18};\tiny{UBISOFT}}      
      \end{lpic}  
    \caption{Left tail of $F_A$.}    
      \end{subfigure}
      \begin{subfigure}[h]{0.5\textwidth}
      \centering
      \begin{lpic}{fb_left_4s(0.4)}    
       \lbl[t]{{90,5};$x$}
       \lbl[t]{{3,60};\begin{turn}{90}\small{CDF}\end{turn}}
      \lbl[t]{{59,33};\tiny{ASML left tail $F_B$}}   
      \lbl[t]{{56,28};\tiny{SAINT GOBAIN}}   
       \lbl[t]{{45,23};\tiny{SIGNIFY}}   
          \lbl[t]{{45,18};\tiny{UBISOFT}}    
      \end{lpic}
     \caption{Left tail of $F_B$.}
      \end{subfigure}
   \caption{\label{fig:tails_fafb_4s} Log-log plots of the tails of the order placement distributions for 4 selected stocks (ASML Holding NV, Compagnie de Saint Gobain SA, Signify NV, Ubisoft Entertainment SA ). The $x$-axes show the number of ticks above (for the right tail) or below (for the left tail) the reference price $x_0$.}
\end{figure}
 Let us first focus on the right tails, \ie\ the upper panels (a) and (b) of figure \ref{fig:tails_fafb_4s}. The plots of the right tails of $F_A$ show apparent power law behaviour in the range between 10 and 1000 ticks above the reference price. After circa 1000 ticks the tails decay faster for a while, but starting around 5000 ticks a new part of the distribution seems to start. The plot is cut-off at 10~000 ticks, but some even reach until 100~000 ticks. These extremes do not contribute to price formation in the auction at all. We focus on the interval of the price axis where price formation occurs: the intersection of the supports of $F_A$ and $F_B$. For the right tail that means $F_B$ provides the effective upper bound (note that the closing price can never take a value above the highest buy order). The support of $F_B$ ranges until around 1000-2000 ticks above the reference price so that is the region we use in our analysis, roughly in line with the intraday results from \cite{ZovkoFarmer02}\footnote{The sell orders (far) above this region can be thought of as coming from another distribution describing patient sellers not relevant to the auction result. To sketch how irrelevant those orders are: the tick size of a stock is normally between 1 and 5 basis points. Assuming a tick size of 2.5 basis points, 2000 ticks correspond to a return of 50\%, while a closing auction return in the order of 1\% is already high.}. Power law behaviour is less clear for $F_B$, but in the range of 100 until 1000 ticks power law behaviour can be recognized for the liquid stocks ASML and Saint Gobain. For the less liquid stocks Signify and Ubisoft it stops earlier around 500 ticks, but this can also be due to smaller volumes of available data. The lower panels (c) and (d) of figure \ref{fig:tails_fafb_4s} show the  left tails of the order placement distributions. These are very similar to the right tails, when the roles of $F_A$ and $F_B$ are switched. Also, on the left side there is a real cut-off point, corresponding to price 0, which is found somewhere between 2000 and 10~000 ticks.
\begin{figure}[htp]
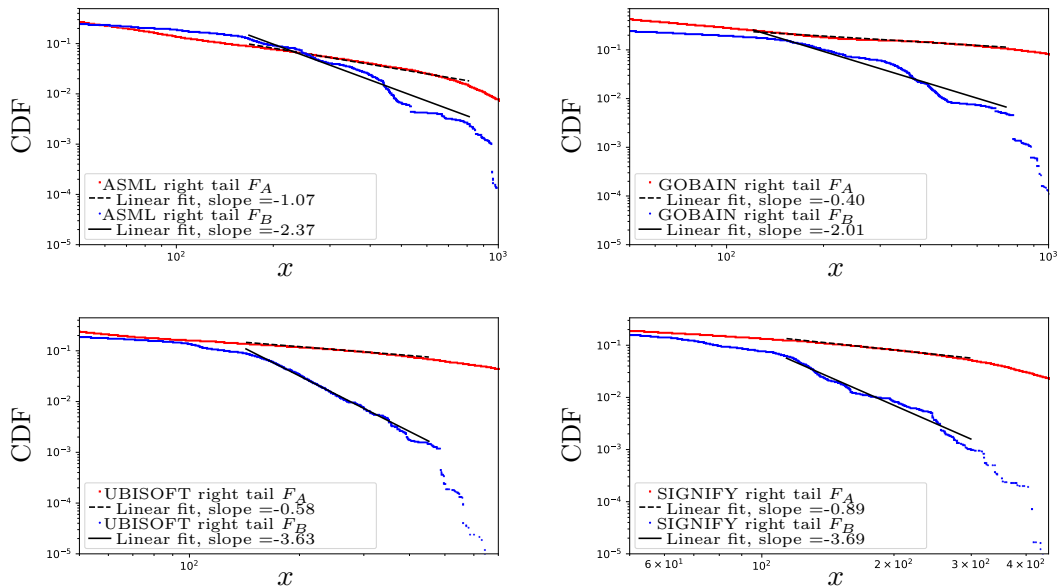

  \centering
      \begin{lpic}{both_right_ASML(0.4)}    
       \lbl[t]{{90,5};$x$}
       \lbl[t]{{3,60};\begin{turn}{90}\small{CDF}\end{turn}}
      \lbl[t]{{58,33};\tiny{ASML right tail $F_A$}}   
      \lbl[t]{{67,28};\tiny{Linear fit, slope =-1.07}}   
       \lbl[t]{{58,23};\tiny{ASML right tail $F_B$}}        
      \lbl[t]{{67,18};\tiny{Linear fit, slope =-2.37}}         
      \end{lpic}
            \begin{lpic}{both_right_GOBAIN(0.4)}    
       \lbl[t]{{90,5};$x$}
       \lbl[t]{{3,60};\begin{turn}{90}\small{CDF}\end{turn}}
      \lbl[t]{{64,33};\tiny{GOBAIN right tail $F_A$}}   
      \lbl[t]{{67,28};\tiny{Linear fit, slope =-0.40}}   
       \lbl[t]{{64,23};\tiny{GOBAIN right tail $F_B$}}        
      \lbl[t]{{67,18};\tiny{Linear fit, slope =-2.01}}         
      \end{lpic}
            \begin{lpic}{both_right_UBISOFT(0.4)}    
       \lbl[t]{{90,5};$x$}
       \lbl[t]{{3,60};\begin{turn}{90}\small{CDF}\end{turn}}
      \lbl[t]{{64,33};\tiny{UBISOFT right tail $F_A$}}   
      \lbl[t]{{67,28};\tiny{Linear fit, slope =-0.58}}   
       \lbl[t]{{64,23};\tiny{UBISOFT right tail $F_B$}}        
      \lbl[t]{{67,18};\tiny{Linear fit, slope =-3.63}}         
      \end{lpic}
                  \begin{lpic}{both_right_SIGNIFY(0.4)}    
       \lbl[t]{{90,5};$x$}
       \lbl[t]{{3,60};\begin{turn}{90}\small{CDF}\end{turn}}
      \lbl[t]{{64,33};\tiny{SIGNIFY right tail $F_A$}}   
      \lbl[t]{{67,28};\tiny{Linear fit, slope =-0.89}}   
       \lbl[t]{{64,23};\tiny{SIGNIFY right tail $F_B$}}        
      \lbl[t]{{67,18};\tiny{Linear fit, slope =-3.69}}         
      \end{lpic}
   \caption{\label{fig:zoomed_right_tails_fafb_4s} Log-log plots of the right tails of the order placement distributions for 4  stocks (ASML Holding NV, Compagnie de Saint Gobain SA, Signify NV, Ubisoft Entertainment SA). The $x$-axes show the number of tick sizes above the reference price $x_0$. Linear fits are also  plotted, fitted on the 0.05-quantile of $F_B$ until the 0.001-quantile of $F_B$, to estimate $a_B$ and $a_A$ }
\end{figure}
In figure \ref{fig:zoomed_right_tails_fafb_4s} we zoom in on the right  tails of $F_B$ and $F_A$ until around 1000 tick sizes above the reference price and provide linear fits as estimators for the values of $a_A$ and $a_B$ (the tail exponents of $F_A$ and $F_B$ as in assumption \ref{ass:fa_heavier_pl}). We perform linear least square fits on the log-log plots of the tails of $F_B$, starting at its 0.05-quantile. Visual inspection shows that in the extreme tails, available data points are too sparse to form a coherent picture. So we stop the fit at the 0.001-quantile of $F_B$, which seems reasonable when inspecting the plots and we make fits for $F_A$ on the same interval\footnote{These choices are somewhat arbitrary, but cut-off choices need to be made in any practical tail analysis (see \cite{Newman05}) and moreover, results do not change substantially when we extend the fit to \eg\ the 0.0001-quantile, or \eg\ start the fit at the 0.01-quantile.}.
For example for ASML, we obtain  $a_A \approx 1.07, ~a_B \approx 2.37$, fitted on the interval of 168 until 862 tick sizes. For all four stocks, $F_A$ shows a  straight, slowly decaying line, resembling a power law with exponents around or even below 1. Furthermore, the tails of $F_B$ decay faster than the tails of $F_A$, with exponents between 2 and 4 (more results are discussed in section \ref{subsec:more_stocks}).

\subsection{Tails of closing auction return distributions} \label{sec:returns}
 For every stock $i$ and day $1 \leq d \leq n$ we have a closing auction return $X_{i,d}$, defined as
$$X_{i,d} = \log(C^{i,d}) - \log(x_0^{i,d}),$$
where $C^{i,d}$ is the closing price of stock $i$ on day $d$ and $x_0^{i,d}$ is the reference price of stock $i$ on day $d$.
Following \eg\  \cite{Gopikrishnanetal98} we standardize the returns per stock. That is, we divide for every stock $i$ the return sample $\{X_{i,d}: 1 \leq d \leq n\}$ by its standard deviation and obtain  a sample of standardized returns of size $n \approx 500$. These samples are all merged together into one large sample to study the tails of the closing auction return distributions.
In figure \ref{fig:closingreturns} the right and left tails of the return distribution are shown in log-log plots, showing clear power law behaviour from 2 up to 10 standard deviations for both tails. Linear least square fits are also shown (starting the fit at 2 standard deviations), giving tail exponents $a=5.28$ for the left tail and $a=4.74$ for the right tail. 
\begin{figure}[h]
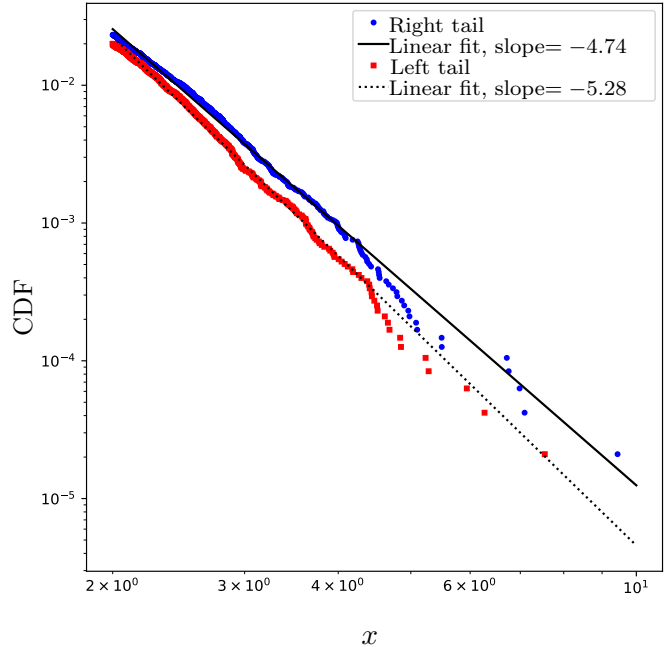

  \centering
      \begin{lpic}{right_left_tails_100s_long(0.55)}   
              \lbl[t]{{90,5};$x$}
       \lbl[t]{{7,90};\begin{turn}{90}\small{CDF}\end{turn}}
      \lbl[t]{{106,153};\scriptsize{Right tail}}
      \lbl[t]{{123,148};\scriptsize{Linear fit, slope$=-4.74$}}  
      \lbl[t]{{104.5,143};\scriptsize{Left tail}}
      \lbl[t]{{123,138};\scriptsize{Linear fit, slope$=-5.28$}}   
      \end{lpic}
      
   \caption{\label{fig:closingreturns} Log-log plot of the tails of the closing auction return distribution for all 100 stocks in our sample. Blue dots show the right tail, that is $P(X>x)$, red squares the left tail, that is $P(X<-x)$, the $x$-axis is in standardized returns. Linear fits are also plotted, giving a tail exponent of $a=4.74$ for the right tail and $a=5.28$ for the left tail.
   }
\end{figure}

This suggests closing auction returns are less heavy tailed than intraday returns over short time intervals, for which a tail exponent $a \approx 3$ is widely supported in the literature (see \eg\ \cite{Gopikrishnanetal98}). This difference might be explained in qualitative terms by the large transacted volumes in the closing auctions. It is known that tails of return distributions become thinner when longer time intervals are considered, an effect that is known as \emph{aggregational Gaussianity} (the empirical fact that return distributions converge to normal distributions when the interval length increases, see \eg\ \cite{Cont01}). This is theoretically supported by the call auction model: the clearing price distribution approaches a normal distribution, when the number of orders tends to infinity (see \cite{Derksenetal20a}, theorem 3.1). Moreover, the empirical effect is known to be stronger if time intervals are measured in trade time \citep{Chakrabortietal11}. In Europe nowadays around 30\% of the daily volume is transacted in the closing auction, which makes the duration of the closing auction in trade time similar to approximately half a day of continuous trading\footnote{The fraction of daily transacted volume that is transacted in closing auctions has increased greatly over the past years, especially since the introduction of MiFID II, see \cite{Derksenetal20b}.}.

\subsection{The effect of market orders} \label{subsec:cvalues}
Before we study the influence of market orders on the tail behaviour of closing auction return distributions, we first investigate the relation between the market order imbalance and the limit order imbalance. In figure \ref{fig:cvalues} the market order imbalance $\Delta=M_B-M_A$ in every closing auction is plotted against the limit order imbalance $N_A-N_B$ in that closing auction, for the four stocks that were also studied in section \ref{sec:empirics_fafb}.
 The figure shows that the proportionality relation between $\Delta$ and $N_A-N_B$ introduced in equation (\ref{eq:delta}) holds approximately, with values of $c$ in the range 0.2-0.4, estimated using linear least square regression. This means that \emph{limit order imbalance is generally in the opposite direction of market order imbalance.}
An explanation for this lies in the chronology of the closing auction. We observe in auction data that the vast majority of market orders is submitted in the first seconds of the closing auction, revealing the market order imbalance early in the auction (during the accumulation phase of the auction, information on the imbalance and an indicative price is released, so it is possible to act on this information). Subsequently, limit orders are placed against the direction of the market order imbalance, reflecting strategic behaviour: when there is a large positive market order imbalance (more buy market orders than sell market orders), one can submit a (possibly large) sell order without adversely affecting the price. 
\begin{figure}[h]
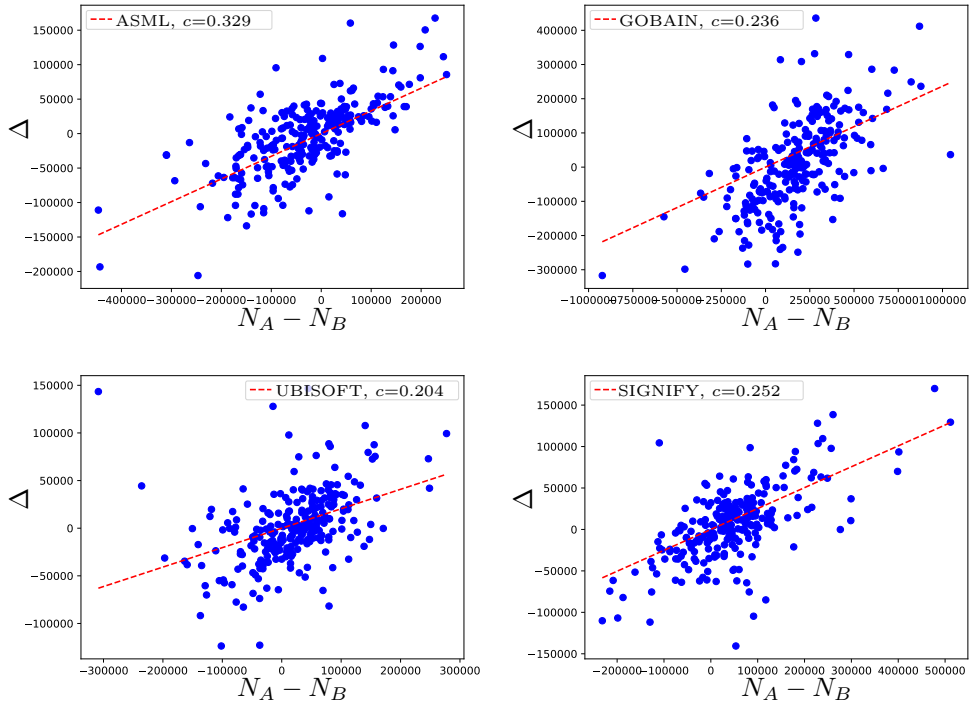

  \centering
      \begin{lpic}{cvalue_ASML(0.4)}    
       \lbl[t]{{0,70};\begin{turn}{90}\small{$\Delta$}\end{turn}}
       \lbl[t]{{90,7};\small{$N_A-N_B$}}
      \lbl[t]{{54,104};\tiny{ASML, $c$=0.329}}        
      \end{lpic}
            \begin{lpic}{cvalue_GOBAIN(0.4)}    
       \lbl[t]{{0,70};\begin{turn}{90}\small{$\Delta$}\end{turn}}
       \lbl[t]{{90,7};\small{$N_A-N_B$}}
      \lbl[t]{{58,104};\tiny{GOBAIN, $c$=0.236}}              
      \end{lpic}
            \begin{lpic}{cvalue_UBISOFT(0.4)}    
       \lbl[t]{{0,70};\begin{turn}{90}\small{$\Delta$}\end{turn}}
       \lbl[t]{{90,7};\small{$N_A-N_B$}}
      \lbl[t]{{112,104};\tiny{UBISOFT, $c$=0.204}}             
      \end{lpic}
                \begin{lpic}{cvalue_SIGNIFY(0.4)}    
       \lbl[t]{{0,70};\begin{turn}{90}\small{$\Delta$}\end{turn}}
       \lbl[t]{{90,7};\small{$N_A-N_B$}}
      \lbl[t]{{58,104};\tiny{SIGNIFY, $c$=0.252}}         
      \end{lpic}
   \caption{\label{fig:cvalues} The difference $N_A-N_B$ plotted against the market order imbalance $\Delta$, showing limit order imbalance goes against the direction of market order imbalance. The dashed red line is the result of linear least square regression, to estimate the value of $c$ in equation (\ref{eq:delta}), which is the slope of the dashed red line (outliers of more than four standard deviations away from the mean are removed).}
\end{figure}

Next, we will investigate the effect of market orders on the tail exponents. Consider figure \ref{fig:example_auction}, where two auction results are shown. Supply and demand curves are represented by the solid lines and the point of intersection is the closing price, indicated by the black star. When market orders are removed, translated supply and demand curves (plotted by the dashed lines) lead to an alternative closing price, represented by the black square. The upper panel shows a situation in which a large positive closing auction return is caused by a high market order imbalance. When the market order imbalance would be removed, the closing price would be much lower (black square). The lower panel shows a very different situation: a small positive closing auction return, but a strongly negative market order imbalance. If in this case the market order imbalance would be removed, the closing auction return would get much higher (black square).
\begin{figure}
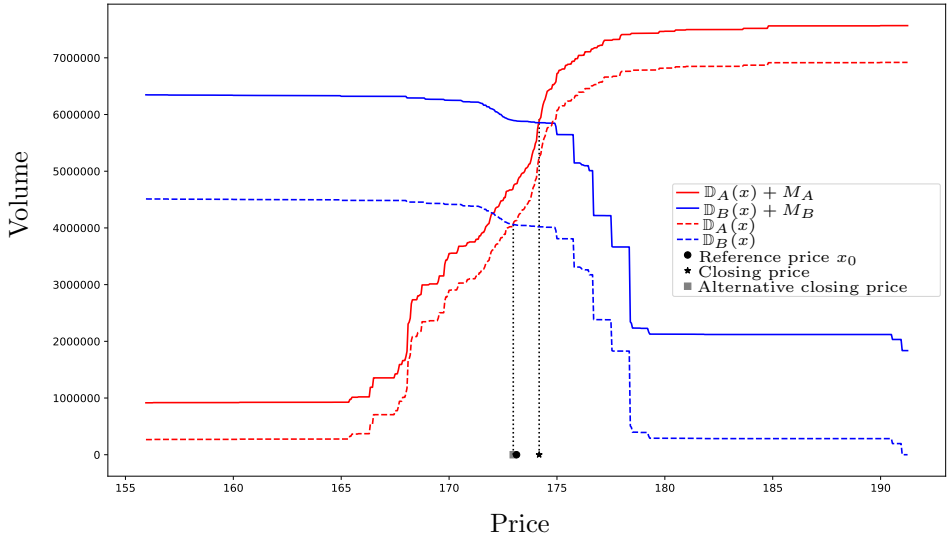
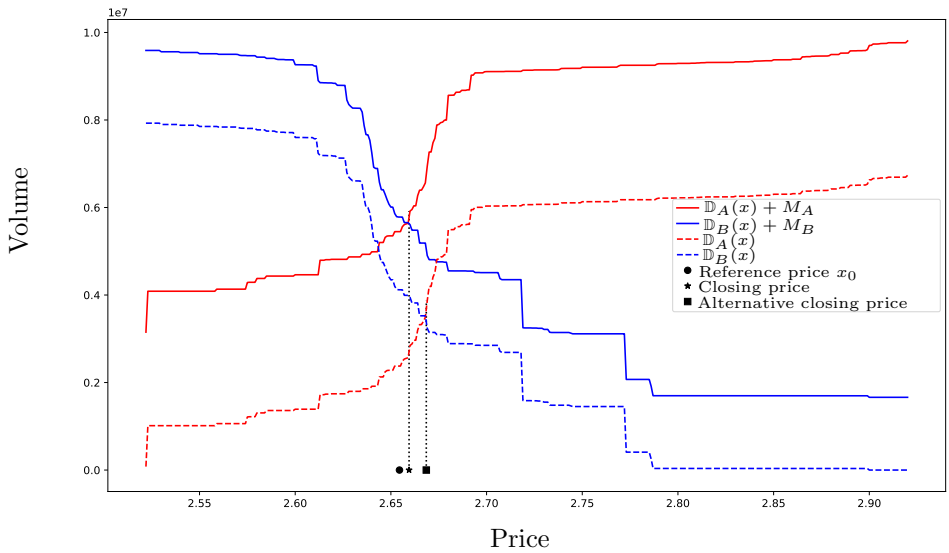

  \begin{subfigure}{\textwidth}
  \centering
          \begin{lpic}{atypical_auction_ASML_2018-03-16(0.4)}    
       \lbl[t]{{180,10};\small{Price}}
       \lbl[t]{{15,130};\begin{turn}{90}\small{Volume}\end{turn}}
       \lbl[t]{{259,119};\tiny{$\DA(x) + M_A$}}
       \lbl[t]{{259,113};\tiny{$\DB(x) + M_B$}}
       \lbl[t]{{249,108};\tiny{$\DA(x)$}}
       \lbl[t]{{249,103};\tiny{$\DB(x)$}}      
       \lbl[t]{{265,97};\tiny{Reference price $x_0$}}      
       \lbl[t]{{257,92};\tiny{Closing price}}      
       \lbl[t]{{273,87};\tiny{Alternative closing price}}      

      \end{lpic}
      \caption{Closing auction ASML, 2018-03-16.}
      \end{subfigure}
      ~
      \begin{subfigure}[h]{\textwidth}
      \centering
      \begin{lpic}{typical_auction_KPN_2018-02-07(0.4)}    
       \lbl[t]{{180,10};\small{Price}}
       \lbl[t]{{15,130};\begin{turn}{90}\small{Volume}\end{turn}}
       \lbl[t]{{259,119};\tiny{$\DA(x) + M_A$}}
       \lbl[t]{{259,113};\tiny{$\DB(x) + M_B$}}
       \lbl[t]{{249,108};\tiny{$\DA(x)$}}
       \lbl[t]{{249,103};\tiny{$\DB(x)$}}      
       \lbl[t]{{265,97};\tiny{Reference price $x_0$}}      
       \lbl[t]{{257,92};\tiny{Closing price}}      
       \lbl[t]{{273,87};\tiny{Alternative closing price}}      

      \end{lpic}
      \caption{Closing auction KPN, 2018-02-07.}
      \end{subfigure}
   \caption{\label{fig:example_auction} Two closing auction results. Solid lines are the supply (red) and demand (blue) curves of the particular closing auction, including market orders (for convenience sell (buy) market orders are placed just below (above) the lowest sell (highest buy) limit order). Dashed lines show the supply and demand curves without market orders. The black dot denotes the reference price $x_0$, the black star denotes the closing price and the black square denotes the alternative price when only limit orders are considered. }
\end{figure}

 The two scenarios presented in figure \ref{fig:example_auction} raise the question which is more common: are large closing auction returns caused by large market order imbalances or is this potential effect cancelled by limit order imbalance and are limit orders usually the driver of large returns? To answer this question, we also investigate the tails of the return distribution of the \emph{alternative closing price}, defined as the intersection point of the supply and demand curves when the market orders are removed (black squares in figure \ref{fig:example_auction}). So, for every stock $i$ and day $d$ we have an alternative closing auction return $\tilde{X}_{i,d}$, defined as
$$\tilde X_{i,d} = \log(\tilde{C}^{i,d}) - \log(x_0^{i,d}),$$
where $\tilde C_{i,d}$ is the alternative closing price of stock $i$ on day $d$. We again standardize these returns per stock, giving for every stock around 500 alternative closing auction returns, which are merged to study the tails. In figure \ref{fig:closingreturns_alt} the tails of the alternative closing price return distribution are shown, together with the tails of the real closing price return distribution from figure \ref{fig:closingreturns}. The figure shows that the tails become \emph{heavier} when market orders are removed. For the right tail we document a tail exponent $a=3.75$ without market orders, compared to $a=4.74$ with market orders. For the left tail, the tail exponent becomes $a=3.9$ when market orders are removed, compared to the value $a=5.28$ when market orders are included. 
\begin{figure}[h]
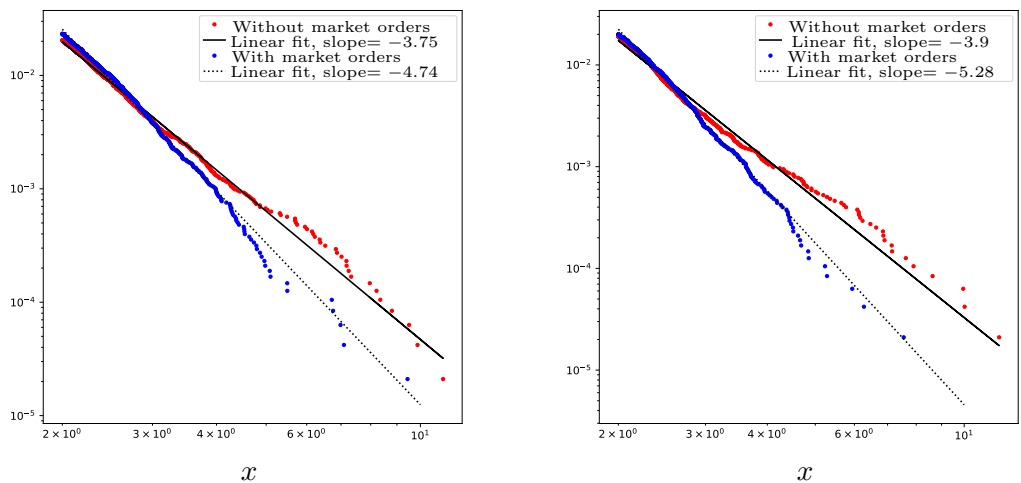

\begin{subfigure}[h]{0.5\textwidth}
  \centering
      \begin{lpic}{right_tails_mos_100s(0.4)}    
       \lbl[t]{{90,5};$x$}
      \lbl[t]{{117,153};\tiny{Without market orders}}
      \lbl[t]{{118,148};\tiny{Linear fit, slope$=-3.75$}}  
      \lbl[t]{{112,143};\tiny{With market orders}}
      \lbl[t]{{118,138};\tiny{Linear fit, slope$=-4.74$}}   
      \end{lpic} \caption{Right tail return distribution, $P(X>x)$}
      \end{subfigure} 
      \begin{subfigure}[h]{0.5\textwidth}
       \begin{lpic}{left_tails_mos_100s(0.4)}    
       \lbl[t]{{90,5};$x$}
      \lbl[t]{{117,153};\tiny{Without market orders}}
      \lbl[t]{{118,148};\tiny{Linear fit, slope$=-3.9$}}  
      \lbl[t]{{112,143};\tiny{With market orders}}
      \lbl[t]{{118,138};\tiny{Linear fit, slope$=-5.28$}}   
      \end{lpic}
      \caption{Left tail return distribution, $P(X<-x)$}
      \end{subfigure}
   \caption{\label{fig:closingreturns_alt} Log-log plot of the tails of the closing auction return distribution for all 100 stocks in the sample. Blue dots show the tails for the real closing auction return distribution, red dots the tails for the alternative closing auction returns that emerge when market orders are removed. 
   }
\end{figure}

It is thus concluded that large closing price fluctuations are in general not caused by a large market order imbalance (at least, not directly). The explanation for this counter-intuitive result lies in the chronology of the auction and the placement of limit orders: when the market order imbalance is positive (negative), there are more sell (buy) limit orders submitted (\cf\ figure \ref{fig:cvalues}). Theorems \ref{thm:tail_x_heavier_fa} and \ref{thm:tailswithmos} give the model's view on the matter and state that without market orders the tail exponent is equal to $a_A+a_B$ and with market orders it is equal to $\min(\frac{(c+1 )a_A}{c},a_A+2a_B)$. This means that tails get heavier without market orders, whenever
\begin{align} \label{eq:heavier_condition} c\leq \frac{a_A}{a_B}.
\end{align}
This equation in fact resembles two conditions that should be fulfilled to make it possible that tails are heavier without market orders (see also equation (\ref{eq:heavier_condition0})). First, $c$ should be small and positive, reflecting that the abovementioned strategic behaviour is strong: when there is a large market order imbalance, in general the limit order difference overcompensates for this. Second, $a_B$ should not be too large compared to $a_A$. This is a condition on the right tail of the buy limit order distribution. Without market orders, the highest buy limit order serves as an upper bound for the closing price. So to obtain heavier tails without market orders, the right tail of $F_B$ should be sufficiently heavy (small $a_B$). It turns out that condition (\ref{eq:heavier_condition}) is indeed satisfied for most of the stocks: for example, for ASML  we obtained estimators $a_B \approx 2.37, a_A \approx 1.07$, $c\approx 0.329$ (\cf\ figures \ref{fig:zoomed_right_tails_fafb_4s} and \ref{fig:cvalues}), satisfying the condition in equation (\ref{eq:heavier_condition}). Indeed, theorems \ref{thm:tail_x_heavier_fa} and \ref{thm:tailswithmos} imply that the tail exponent for closing auction returns of ASML is $a_A+a_B=3.44$ without market orders and $\frac{(c+1 )a_A}{c} =4.32$ with market orders. In the next subsection we will verify the theoretical results on the whole sample consisting of 100 stocks. 
\subsection{Model-predicted and realized tail exponents compared}
 \label{subsec:more_stocks}
In this subsection the relations predicted by the model are tested over the whole sample of 100 stocks. For every stock we estimate the tail exponents of the order placement distributions ($a_A$ and $a_B$) and the value of the parameter $c$ (as in equation (\ref{eq:delta})). The results are shown in table \ref{table:results_small} (for 50 stocks with the lower market capitalizations) and table \ref{table:results_large} (for 50 stocks with the higher market capitalizations). To estimate the parameter $c$, we use linear least squares regression and to estimate the values of $a_A$ and $a_B$ we use the method described in section \ref{sec:empirics_fafb}: for every stock, we make linear least square fits on double logarithmic plots as in figure \ref{fig:zoomed_right_tails_fafb_4s}, on the interval between the 0.05- and 0.001-quantiles of $F_B$. The absolute values of the resulting slopes are the estimators for $a_A$ and $a_B$. For example, for ASML we obtain in this way estimators $a_B\approx 2.37,~a_A \approx 1.07$ and for Ubisoft we find $a_B \approx 3.63,~ a_A \approx 0.58$, \cf\ figure \ref{fig:zoomed_right_tails_fafb_4s}. In tables \ref{table:results_small} and \ref{table:results_large} the results are shown for all stocks in the sample, the columns $a_B(r)$ and $a_A(r)$ give the estimated tail exponents for the right tails of $F_B$ and $F_A$.
For the left tails, the same method applies when the roles of $F_B$ and $F_A$ are interchanged. On the left side, $F_B$ has a heavier tail and $F_A$ provides the effective lower bound. 
\begin{figure}[h]
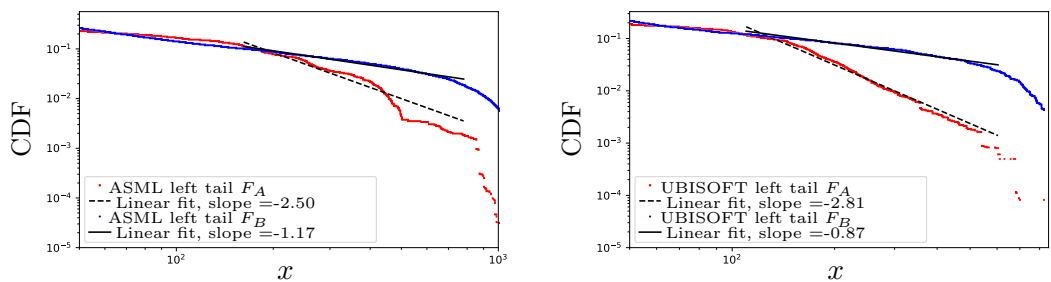

  \centering
      \begin{lpic}{both_left_ASML(0.4)}    
       \lbl[t]{{90,5};$x$}
       \lbl[t]{{3,60};\begin{turn}{90}\small{CDF}\end{turn}}
      \lbl[t]{{58,33};\tiny{ASML left tail $F_A$}}   
      \lbl[t]{{67,28};\tiny{Linear fit, slope =-2.50}}   
       \lbl[t]{{58,23};\tiny{ASML left tail $F_B$}}        
      \lbl[t]{{67,18};\tiny{Linear fit, slope =-1.17}}         
      \end{lpic}
                  \begin{lpic}{both_left_UBISOFT(0.4)}    
       \lbl[t]{{90,5};$x$}
       \lbl[t]{{3,60};\begin{turn}{90}\small{CDF}\end{turn}}
      \lbl[t]{{64,33};\tiny{UBISOFT left tail $F_A$}}   
      \lbl[t]{{67,28};\tiny{Linear fit, slope =-2.81}}   
       \lbl[t]{{64,23};\tiny{UBISOFT left tail $F_B$}}        
      \lbl[t]{{67,18};\tiny{Linear fit, slope =-0.87}}         
      \end{lpic}
   \caption{\label{fig:zoomed_left_tails_fafb_2s} Log-log plots of the left tails of the order placement distributions for ASML Holding NV and Ubisoft Entertainment SA. The $x$-axis shows the number of tick sizes above the reference price $x_0$. Linear fits are also  plotted, fitted on the 0.05-quantile of $F_A$ until the 0.001-quantile of $F_A$, to estimate $a_B$ and $a_A$ }
\end{figure}

 \begin{figure}
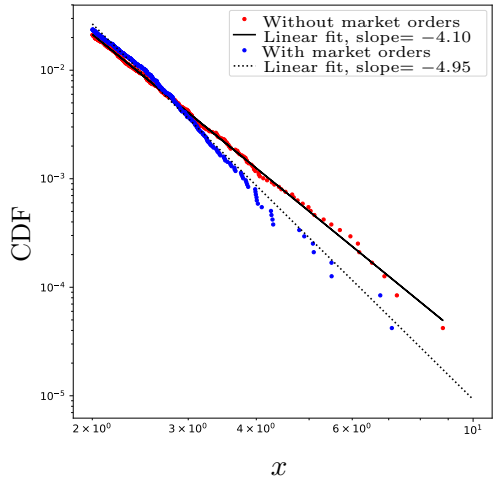
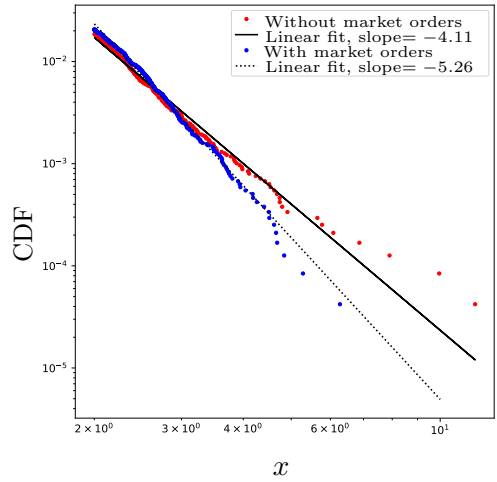
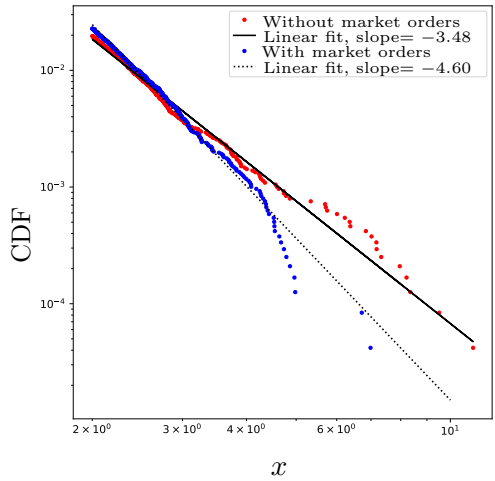
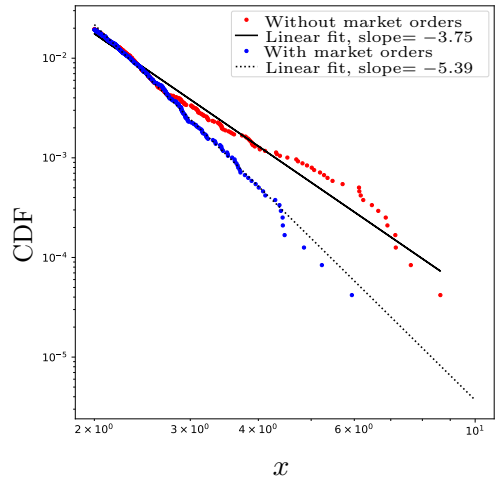

\begin{subfigure}[h]{0.5\textwidth}
  \centering
      \begin{lpic}{smallcaps_right_long(0.4)}    
       \lbl[t]{{90,5};$x$}
       \lbl[t]{{5,90};\begin{turn}{90}\small{CDF}\end{turn}}
      \lbl[t]{{117,153};\tiny{Without market orders}}
      \lbl[t]{{119,148};\tiny{Linear fit, slope$=-4.10$}}  
      \lbl[t]{{112,143};\tiny{With market orders}}
      \lbl[t]{{119,138};\tiny{Linear fit, slope$=-4.95$}}   
      \end{lpic} \caption{Small caps, right tail return distribution.}
      \end{subfigure} 
      \begin{subfigure}[h]{0.5\textwidth}
       \begin{lpic}{smallcaps_left_long(0.4)}    
       \lbl[t]{{90,5};$x$}
       \lbl[t]{{5,90};\begin{turn}{90}\small{CDF}\end{turn}}
      \lbl[t]{{117,153};\tiny{Without market orders}}
      \lbl[t]{{119,148};\tiny{Linear fit, slope$=-4.11$}}  
      \lbl[t]{{112,143};\tiny{With market orders}}
      \lbl[t]{{119,138};\tiny{Linear fit, slope$=-5.26$}}   
      \end{lpic}
      \caption{Small caps, left tail return distribution.}
      \end{subfigure}
   
\begin{subfigure}[h]{0.5\textwidth}
  \centering
      \begin{lpic}{largecaps_right_long(0.4)}    
       \lbl[t]{{90,5};$x$}
       \lbl[t]{{5,90};\begin{turn}{90}\small{CDF}\end{turn}}
      \lbl[t]{{117,153};\tiny{Without market orders}}
      \lbl[t]{{119,148};\tiny{Linear fit, slope$=-3.48$}}  
      \lbl[t]{{112,143};\tiny{With market orders}}
      \lbl[t]{{119,138};\tiny{Linear fit, slope$=-4.60$}}   
      \end{lpic} \caption{Large caps, right tail return distribution.}
      \end{subfigure} 
      \begin{subfigure}[h]{0.5\textwidth}
       \begin{lpic}{largecaps_left_long(0.4)}    
       \lbl[t]{{90,5};$x$}
       \lbl[t]{{5,90};\begin{turn}{90}\small{CDF}\end{turn}}
      \lbl[t]{{117,153};\tiny{Without market orders}}
      \lbl[t]{{119,148};\tiny{Linear fit, slope$=-3.75$}}  
      \lbl[t]{{112,143};\tiny{With market orders}}
      \lbl[t]{{119,138};\tiny{Linear fit, slope$=-5.39$}}   
      \end{lpic}
      \caption{Large caps, left tail return distribution.}
      \end{subfigure}
   \caption{\label{fig:closingreturns_small_largecaps} Log-log plots of the tails of the closing auction return distributions for the 50 small cap stocks of table \ref{table:results_small} (upper panel) and 50 large cap stocks of table \ref{table:results_large} (lower panel). Blue dots show the tails for the real closing auction return distribution, red dots the tails for the alternative closing auction returns that emerge when market orders are removed.
   }
\end{figure}

In figure \ref{fig:zoomed_left_tails_fafb_2s} the left tails of the order placement distributions are shown for ASML and Ubisoft, as well as the linear least square fits, showing that for the left tails $a_A\approx 2.50, ~a_B \approx 1.17$ for ASML and $a_A\approx 2.81, ~a_B \approx 0.87$ for Ubisoft. In tables \ref{table:results_small} and \ref{table:results_large} the columns $a_B(l)$ and $a_A(l)$ give the estimated tail exponents for the left tails of $F_B$ and $F_A$.
Estimates for $a_A,~a_B$ and $c$ give rise to an estimate for the tail exponent $\underline a$ for the return distribution of that particular stock. With market orders $\underline a = \min(\frac{(c+1 )a_A}{c},a_A+2a_B)$  (\cf\ theorem \ref{thm:tailswithmos})  and without market orders $\underline a = a_A + a_B$ (\cf\ theorem \ref{thm:tail_x_heavier_fa})\footnote{Note that for the left tails the roles of $a_A$ and $a_B$ need to be interchanged (see also remark \ref{remark:clearingprices}).}. Ideally, we would test these predictions against the realized tail exponents of the return distribution \emph{for every stock}. However, as noted in the beginning of this section, this is not possible, because we only have around 500 closing auction returns per stock. Instead, we can verify the predictions over groups of stocks, by comparing estimated tail coefficients with the model's average predicted values.

First, consider the whole sample of 100 stocks. In figure \ref{fig:closingreturns_alt} it was shown that the right tail of the closing price return distribution has an estimated tail exponent of $a=4.74$, which changes to $a=3.75$ if market orders are removed. If we take the average of the model's predictions over all 100 stocks, we find an average predicted tail exponent of 4.89 with market orders (column `$\underline a(r)$ MO' in tables \ref{table:results_small} and \ref{table:results_large}) and 3.89 without market orders (column `$\underline a(r)$ no MO' in tables \ref{table:results_small} and \ref{table:results_large}). Furthermore, figure \ref{fig:closingreturns_alt} shows that the left tail of the closing price return distribution has an estimated tail exponent of $a=5.28$, which changes to $a=3.90$ if the market orders are removed. For the left tail, the average predicted tail exponent over all 100 stocks equals 5.01 with market orders (column `$\underline a(l)$ MO' in tables \ref{table:results_small} and \ref{table:results_large}) and 3.72 without market orders (column `$\underline a(l)$ no MO' in tables \ref{table:results_small} and \ref{table:results_large}).  
The predicted tail exponents vary a lot between the different stocks, suggesting that the heaviness of the tails depends on the stock. To additionally test if these per stock predictions give information about the real tail exponents, we split our sample into 50 stocks with the lower market caps (those in table \ref{table:results_small}) and 50 stocks with the higher market caps (table \ref{table:results_large}). In that way, the groups are kept large enough to examine the tails of the closing auction return distributions. 

In figure \ref{fig:closingreturns_small_largecaps} the tails of the closing auction return distribution for the 50 small caps and the 50 large caps are shown in double logarithmic plots, again with and without market orders (similar to figure \ref{fig:closingreturns_alt}). The linear fits to the double logarithmic plots are the realized tail exponents for the both groups, which can again be compared to the average predicted values in tables \ref{table:results_small} and \ref{table:results_large}. The results are summarized in table \ref{table:predicted_vs_realized}, showing first of all that the model's predicted exponents are quite close to the realized exponents. Given that estimation of tails (and tail exponents in particular) is generally thought of as a difficult statistical problem, the congruence is quite remarkable.
Second, based on the modelling assumption in equation (\ref{eq:delta}), the model predicts correctly that the tails get heavier if market orders are removed, and by how much.
The theoretical predictions are especially accurate for the case without market orders, which is not surprising: theorem \ref{thm:tail_x_heavier_fa} holds very generally and follows directly from the mechanics of the closing auction. For the case with market orders, more assumptions were made (see assumption \ref{ass:delta_general}). Most importantly, we assumed  equation (\ref{eq:delta}) holds true, which of course in reality holds only approximately (see also figure \ref{fig:cvalues}). When looking at tables \ref{table:results_small} and \ref{table:results_large}, the predictions for the case with market orders vary strongly between the stocks. We do not claim that the most extreme values that are predicted are close to reality, but we have shown that, on average, model predicted and realized tail exponents match remarkably well.

\begin{table}[htb]
\scriptsize
\centering
\begin{tabular}{l ll|ll l ll|ll}
  &\multicolumn{4}{c}{Left tails}  &&\multicolumn{4}{c}{Right tails}  \\
  \hline
           & \multicolumn{2}{c}{MO}  &  \multicolumn{2}{c}{No MO}     &  ~     & \multicolumn{2}{c}{MO}   &  \multicolumn{2}{c}{No MO}          
        \\   \hline   
           & Predicted   & Realized & Predicted & Realized & & Predicted   & Realized & Predicted & Realized \\
All stocks & 5.01       & 5.28     & 3.72      & 3.90  & & 4.89        & 4.74     & 3.89      & 3.75        \\
Small caps &  5.76        & 5.26     & 4.14      & 4.11 && 5.19        & 4.95     & 4.17      & 4.10      \\
Large caps &  4.25        & 5.39     & 3.29      & 3.75 && 4.59        & 4.60      & 3.61      & 3.48     \\
\hline
\end{tabular}
\caption{\label{table:predicted_vs_realized} Average predicted tail exponents compared to realized tail exponents. Predicted exponents are averages over tables \ref{table:results_small} and \ref{table:results_large}, realized exponents are the results of the linear fits in figures \ref{fig:closingreturns_alt} (all stocks) and \ref{fig:closingreturns_small_largecaps} (small and large caps), for the cases with (MO) and without (No MO) market orders.}
\end{table}

\newpage
\begin{table}[] \tiny
\begin{tabular}{llllllll p{0.75 cm} p{0.75cm} p{0.75 cm} p{0.75cm}}
Stock                     & Exch. & Mcap & $a_A$(l) & $a_B$(l)     & $a_A$(r)    & $a_B$(r)   & $c$     & $\underline a$(l)  no~MO & $\underline a$(l) ~~~MO &$\underline a$(r) no~MO & $\underline a$(r) ~~~MO  \\
\hline
ASM INTL                    & AMS      & 6.7  & 3.312 & 1.154 & 0.717 & 3.377 & 0.127 & 4.466         & 7.778     & 4.094          & 6.357      \\
AALBERTS                    & AMS      & 3.8  & 3.178 & 0.763 & 1.207 & 3.301   & 0.174 & 3.941         & 5.135     & 4.508          & 7.808      \\
WDP REIT                    & BRU      & 5.2  & 3.142 & 3.279 & 1.739 & 3.032 & 0.113 & 6.421         & 9.562     & 4.771          & 7.804      \\
REXEL                       & PAR      & 3.2  & 2.105 & 1.782 & 1.465 & 2.626 & 0.184 & 3.887         & 5.992     & 4.091          & 6.718      \\
EURONEXT                    & PAR      & 6.8  & 3.319 & 1.943 & 1.918 & 3.229 & 0.164 & 5.262         & 8.582     & 5.146          & 8.375      \\
IMCD GROUP                  & AMS      & 6.0  & 3.328 & 0.972 & 1.859 & 3.016 & 0.158 & 4.300           & 7.118     & 4.875          & 7.892      \\
SIGNIFY                     & AMS      & 4.6  & 3.485 & 1.181 & 0.888 & 3.690  & 0.252 & 4.666         & 5.866     & 4.579          & 4.413      \\
ALTEN                       & PAR      & 2.8  & 3.443 & 1.476 & 1.017 & 2.420  & 0.137 & 4.919         & 8.362     & 3.437          & 5.857      \\
BIC                         & PAR      & 1.9  & 3.025 & 1.644 & 0.916 & 3.820  & 0.324 & 4.669         & 6.720      & 4.736          & 3.746      \\
EUTELSAT COM    & PAR      & 1.9  & 2.969 & 0.993 & 0.677 & 3.752 & 0.201   & 3.962         & 5.972     & 4.429          & 4.068      \\
INGENICO GROUP              & PAR      & 8.5  & 1.534 & 0.970  & 0.740  & 2.996 & 0.129 & 2.504         & 4.038     & 3.735          & 6.461      \\
EURAZEO                     & PAR      & 3.3  & 4.512 & 1.953 & 1.536 & 3.928 & 0.134 & 6.464         & 10.976    & 5.464          & 9.393      \\
AEGON                       & AMS      & 5.0  & 1.552 & 0.801 & 0.421 & 2.561 & 0.15  & 2.353         & 3.905     & 2.982          & 3.220       \\
KPN KON                     & AMS      & 10.0 & 2.363 & 1.011 & 0.584 & 3.942 & 0.207 & 3.373         & 5.736     & 4.526          & 3.413      \\
RANDSTAD                    & AMS      & 8.4  & 3.071 & 1.130  & 0.939 & 3.603 & 0.291 & 4.201         & 5.016     & 4.541          & 4.166      \\
KLEPIERRE REIT              & PAR      & 3.5  & 3.242 & 1.016 & 0.503 & 3.317 & 0.371 & 4.259         & 3.754     & 3.820           & 1.858      \\
SUEZ                        & PAR      & 9.9  & 1.416 & 0.866 & 0.724 & 2.588 & 0.250  & 2.281         & 3.697     & 3.312          & 3.613      \\
GALP ENERGIA          & LIS      & 6.8  & 4.467 & 2.173 & 1.438 & 3.977 & 0.484 & 6.640          & 6.663     & 5.415          & 4.410       \\
ARKEMA                      & PAR      & 7.2  & 4.292 & 1.231 & 1.354 & 4.441 & 0.278 & 5.522         & 5.659     & 5.795          & 6.225      \\
COVIVIO                     & PAR      & 5.2  & 3.400   & 3.485 & 2.230  & 3.279 & 0.304 & 6.884         & 10.284    & 5.509          & 8.788      \\
ICADE REIT                  & PAR      & 3.4  & 4.179 & 1.649 & 1.371 & 3.499 & 0.231 & 5.828         & 8.798     & 4.870           & 7.314      \\
IPSEN                       & PAR      & 6.5  & 2.515 & 1.463 & 0.844 & 2.594 & 0.200   & 3.978         & 6.493     & 3.439          & 5.072      \\
ORPEA                       & PAR      & 5.9  & 1.605 & 0.786 & 0.987 & 1.683 & 0.126 & 2.391         & 3.997     & 2.671          & 4.354      \\
SCOR                        & PAR      & 4.5  & 3.363 & 1.475 & 0.941 & 3.944 & 0.378 & 4.837         & 5.372     & 4.885          & 3.427      \\
GETLINK                     & PAR      & 6.4  & 1.910 & 1.341 & 1.166 & 3.294 & 0.148 & 3.251         & 5.161     & 4.460           & 7.753      \\
J.MARTINS SGPS              & LIS      & 9.2  & 4.043 & 1.022 & 1.027 & 3.621 & 0.276 & 5.065         & 4.730      & 4.648          & 4.752      \\
DASSAULT AVIAT              & PAR      & 6.3  & 4.021 & 1.155 & 0.764 & 4.088 & 0.282 & 5.176         & 5.247     & 4.852          & 3.471      \\
EDENRED                     & PAR      & 10.1 & 3.832 & 2.070  & 2.352 & 3.763 & 0.282 & 5.902         & 9.419     & 6.115          & 9.879      \\
PUBLICIS GROUPE             & PAR      & 7.6  & 2.600   & 1.251 & 0.704 & 3.221 & 0.390  & 3.851         & 4.462     & 3.925          & 2.509      \\
ATOS                        & PAR      & 7.5  & 2.323 & 0.834 & 0.499 & 2.449 & 0.267 & 3.157         & 3.959     & 2.948          & 2.368      \\
JCDECAUX                    & PAR      & 2.9  & 3.682 & 1.578 & 1.177 & 3.845 & 0.229 & 5.260          & 8.482     & 5.022          & 6.323      \\
EIFFAGE                     & PAR      & 6.9  & 4.086 & 1.324 & 1.785 & 4.083 & 0.248 & 5.410          & 6.656     & 5.868          & 8.973      \\
GECINA                      & PAR      & 7.8  & 2.988 & 2.394 & 2.004 & 2.928 & 0.321 & 5.382         & 8.370      & 4.932          & 7.860       \\
NATIXIS                     & PAR      & 6.5  & 0.763 & 0.513 & 0.524 & 1.672 & 0.155 & 1.276         & 2.040      & 2.196          & 3.868      \\
SES FDR                        & PAR      & 3.0  & 3.138 & 0.649 & 0.905 & 3.218 & 0.186 & 3.786         & 4.132     & 4.123          & 5.768      \\
SEB                         & PAR      & 7.6  & 3.998 & 1.284 & 0.998 & 3.807 & 0.215 & 5.282         & 7.246     & 4.805          & 5.629      \\
UBISOFT           & PAR      & 10.2 & 2.814 & 0.870  & 0.578 & 3.635 & 0.204 & 3.684         & 5.136     & 4.213          & 3.409      \\
ALSTOM                      & PAR      & 9.4  & 1.988 & 0.715 & 0.866 & 2.955 & 0.279 & 2.703         & 3.280      & 3.821          & 3.975      \\
TECHNIPFMC                  & PAR      & 2.6  & 1.113 & 0.729 & 0.552 & 2.478 & 0.217 & 1.842         & 2.955     & 3.029          & 3.097      \\
ACCOR                       & PAR      & 5.9  & 2.384 & 0.677 & 0.645 & 3.079 & 0.251 & 3.061         & 3.374     & 3.724          & 3.215      \\
VEOLIA    & PAR      & 9.8  & 1.476 & 0.838 & 0.484 & 2.598 & 0.267 & 2.314         & 3.791     & 3.083          & 2.296      \\
COLRUYT               & BRU      & 7.3  & 3.340  & 1.373 & 1.245 & 3.390  & 0.318 & 4.713         & 5.689     & 4.635          & 5.159      \\
AGEAS                       & BRU      & 6.7  & 2.408 & 1.115 & 1.254 & 2.735 & 0.321 & 3.524         & 4.585     & 3.989          & 5.155      \\
SOLVAY                      & BRU      & 8.0  & 1.866 & 0.856 & 0.440  & 1.241 & 0.193 & 2.722         & 4.589     & 1.681          & 2.725      \\
UMICORE                     & BRU      & 8.9  & 2.428 & 0.834 & 0.486 & 2.003 & 0.314 & 3.262         & 3.490      & 2.489          & 2.033      \\
PROXIMUS                    & BRU      & 5.2  & 2.618 & 0.845 & 0.666 & 2.820  & 0.271 & 3.463         & 3.964     & 3.486          & 3.126      \\
ABN AMRO BANK               & AMS      & 6.9  & 1.989 & 1.216 & 0.627 & 3.004 & 0.195 & 3.205         & 5.194     & 3.631          & 3.838      \\
CNP ASSURANCES              & PAR      & 7.3  & 3.899 & 1.615 & 2.046 & 2.756 & 0.381 & 5.514         & 5.855     & 4.802          & 7.418      \\
UNIBAIL RODAMCO  & AMS      & 5.7  & 1.983 & 1.368 & 1.052 & 1.947 & 0.390  & 3.351         & 4.876     & 2.999          & 3.750       \\
SODEXO                      & PAR      & 8.8  & 3.559 & 1.852 & 1.696 & 3.668 & 0.245 & 5.411         & 8.970      & 5.364          & 8.629     \\
\hline 
Average & - & 6.4 & 	2.84 &	1.30 &	1.06 &	3.12 &	0.24 &	4.14 & 5.76 & 4.17 &5.19
\end{tabular}
\caption{\label{table:results_small} Table of results, for the 50 stocks in our sample with the lower market cap. The column Exch. displays the exchange the stock is traded on (Amsterdam, Paris, Brussels or Lisbon)
and the column Mcap shows the market capitalization of the stock in billions of euros (in October 2020). Then, $a_A$ and $a_B$ are the estimated tail exponents of sell and buy limit order distributions, for the left (l) and right (r) tail. $c$ is the estimator for the constant in equation (\ref{eq:delta}) and $\underline a=a_A + a_B$ without market orders (no MO), and $\underline a =\min(\frac{(c+1 )a_A}{c},a_A+2a_B)$ with market orders (MO), both displayed for left (l) and right (r) tails.}
\end{table}

\begin{table}[] \tiny
\begin{tabular}{llllllll p{0.75 cm} p{0.75cm} p{0.75 cm} p{0.75cm}}
Stock                     & Exch. & Mcap & $a_A$(l) & $a_B$(l)     & $a_A$(r)    & $a_B$(r)   & $c$     & $\underline a$(l)  no~MO & $\underline a$(l) ~~~MO &$\underline a$(r) no~MO & $\underline a$(r) ~~~MO  \\
\hline
AMUNDI                    & PAR      & 12.3  & 2.066 & 1.544 & 0.816 & 2.496 & 0.110  & 3.610          & 5.676     & 3.311          & 5.807      \\
BIOMERIEUX ORD               & PAR      & 16.4  & 2.815 & 1.595 & 2.294 & 3.312 & 0.247 & 4.410          & 7.225     & 5.606          & 8.917      \\
NN GROUP                  & AMS      & 10.9  & 3.187 & 1.051 & 1.010  & 3.797 & 0.376 & 4.238         & 3.847     & 4.807          & 3.698      \\
SARTORIUS  & PAR      & 28.9  & 2.928 & 1.069 & 1.449 & 2.474 & 0.280  & 3.997         & 4.892     & 3.923          & 6.397      \\
WORLDLINE                 & PAR      & 13.2  & 2.386 & 1.001 & 1.131 & 2.902 & 0.241 & 3.386         & 5.146     & 4.034          & 5.819      \\
EDP                       & LIS      & 18.0  & 3.236 & 0.710  & 1.043 & 3.683 & 0.341 & 3.946         & 2.793     & 4.727          & 4.106      \\
TELEPERFORMANCE           & PAR      & 16.1  & 3.869 & 1.862 & 1.256 & 4.010  & 0.134 & 5.731         & 9.601       & 5.266          & 9.276      \\
BOUYGUES                  & PAR      & 11.7  & 2.302 & 1.124 & 0.953 & 2.452 & 0.331 & 3.425         & 4.518     & 3.405          & 3.833      \\
AHOLD DEL                 & AMS      & 26.6  & 1.334 & 0.727 & 0.875 & 2.234 & 0.326 & 2.061         & 2.960      & 3.108          & 3.562      \\
AKZO NOBEL                & AMS      & 17.7  & 2.678 & 1.921 & 1.190  & 2.669 & 0.292 & 4.599         & 7.276     & 3.859          & 5.268      \\
ASML HOLDING              & AMS      & 138.3 & 2.497 & 1.165 &      2.369 &   1.073    & 0.329 & 3.662         & 4.704     &               3.442 &    4.334        \\
DSM KON                   & AMS      & 24.9  & 3.283 & 1.333 & 1.158 & 3.381 & 0.312 & 4.616         & 5.601     & 4.539          & 4.866      \\
HEINEKEN                  & AMS      & 45.4  & 3.917 & 1.101 & 1.146 & 3.833 & 0.270  & 5.018         & 5.177     & 4.979          & 5.391      \\
ING GROEP                 & AMS      & 24.7  & 1.344 & 1.067 & 0.419 & 2.333 & 0.139 & 2.411         & 3.755     & 2.752          & 3.428      \\
PHILIPS KON               & AMS      & 37.5  & 3.129 & 0.498 & 0.884 & 3.566 & 0.318 & 3.626         & 2.063     & 4.450           & 3.666      \\
UNILEVER                  & AMS      & 138.7 & 2.929 & 0.619 & 1.417 & 2.879 & 0.322 & 3.548         & 2.543     & 4.297          & 5.818      \\
WOLTERS KLUWER            & AMS      & 19.3  & 3.473 & 1.417 & 1.823 & 2.997 & 0.244 & 4.891         & 7.235     & 4.820          & 7.817      \\
DANONE                    & PAR      & 34.6  & 2.195 & 1.112 & 0.761 & 2.312 & 0.404 & 3.307         & 3.862     & 3.073          & 2.642      \\
BNP PARIBAS ACT.A         & PAR      & 40.2  & 1.326 & 0.771 & 0.636 & 2.133 & 0.223 & 2.098         & 3.424     & 2.769          & 3.496      \\
AXA                       & PAR      & 36.1  & 0.740  & 0.698 & 0.548 & 1.878 & 0.321 & 1.438         & 2.178     & 2.426          & 2.253      \\
SOCIETE GENERALE          & PAR      & 10.2  & 0.815 & 0.923 & 0.383 & 2.051 & 0.160  & 1.738         & 2.553     & 2.434          & 2.777      \\
L'OREAL                   & PAR      & 163.0 & 3.201 & 0.685 & 1.209 & 2.552 & 0.270  & 3.886         & 3.222     & 3.760           & 5.683      \\
SANOFI                    & PAR      & 108.1 & 1.245 & 0.741 & 0.790  & 2.080  & 0.411 & 1.986         & 2.544     & 2.869          & 2.712      \\
SAINT GOBAIN              & PAR      & 20.0  & 1.242 & 0.936 & 0.402 & 2.014 & 0.237 & 2.178         & 3.420      & 2.416          & 2.099      \\
LEGRAND                   & PAR      & 18.6  & 4.179 & 1.304 & 1.790  & 3.851 & 0.328 & 5.482         & 5.283     & 5.641          & 7.254      \\
TOTAL                     & PAR      & 74.8  & 0.592 & 0.632 & 0.549 & 0.534 & 0.308 & 1.224         & 1.816     & 1.082          & 1.616      \\
HEINEKEN HOLDING          & AMS      & 20.2  & 3.198 & 1.244 & 2.004 & 3.506 & 0.333 & 4.442         & 4.982     & 5.509          & 8.025      \\
ESSILORLUXOTTICA          & PAR      & 50.8  & 2.433 & 0.555 & 1.040  & 2.691 & 0.347 & 2.988         & 2.155     & 3.731          & 4.037      \\
AB INBEV                  & BRU      & 93.5  & 1.800   & 0.809 & 0.499 & 2.441 & 0.210  & 2.610          & 4.410      & 2.940       & 2.879      \\
DASSAULT SYSTEM           & PAR      & 41.4  & 3.381 & 1.439 & 1.144 & 2.779 & 0.341 & 4.819         & 5.653     & 3.923          & 4.495      \\
CHRISTIAN DIOR SE         & PAR      & 74.0  & 3.482 & 1.585 & 1.860  & 4.092 & 0.095 & 5.068         & 8.550      & 5.952          & 10.044     \\
ARCELORMITTAL             & AMS      & 13.4  & 0.450  & 0.610 & 0.566 & 2.015 & 0.138 & 1.060          & 1.511     & 2.580           & 4.595      \\
SAFRAN                    & PAR      & 38.7  & 2.944 & 1.467 & 1.154 & 2.380  & 0.386 & 4.411         & 5.267     & 3.534          & 4.141      \\
ENGIE                     & PAR      & 28.3  & 0.485 & 0.652 & 0.269 & 1.846 & 0.138 & 1.136         & 1.621     & 2.114          & 2.220       \\
EDF                       & PAR      & 31.9  & 0.789 & 0.872 & 0.906 & 2.429 & 0.234 & 1.662         & 2.451     & 3.335          & 4.780       \\
CREDIT AGRICOLE           & PAR      & 21.2  & 0.642 & 0.780  & 0.670  & 1.729 & 0.193 & 1.422         & 2.065     & 2.399          & 4.128      \\
CAPGEMINI                 & PAR      & 18.5  & 2.622 & 0.754 & 0.747 & 2.987 & 0.256 & 3.377         & 3.701     & 3.735          & 3.667      \\
AIRBUS                    & PAR      & 50.4  & 2.921 & 1.232 & 1.320  & 3.268 & 0.116 & 4.153         & 7.075     & 4.588          & 7.856      \\
ORANGE                    & PAR      & 25.3  & 0.480  & 0.925 & 0.515 & 2.162 & 0.250  & 1.405         & 1.885     & 2.676          & 2.577      \\
THALES                    & PAR      & 13.9  & 2.656 & 1.241 & 0.865 & 2.351 & 0.317 & 3.897         & 5.161     & 3.216          & 3.598      \\
MICHELIN                  & PAR      & 16.6  & 2.360  & 0.740  & 0.416 & 2.724 & 0.320  & 3.100           & 3.050      & 3.139          & 1.713      \\
KERING                    & PAR      & 73.7  & 2.572 & 0.864 & 0.684 & 2.463 & 0.332 & 3.436         & 3.466     & 3.146          & 2.742      \\
PERNOD RICARD             & PAR      & 37.1  & 2.740  & 1.090  & 1.153 & 2.256 & 0.297 & 3.829         & 4.761     & 3.409          & 5.036      \\
SCHNEIDER ELECTRIC   & PAR      & 58.2  & 2.644 & 1.534 & 1.058 & 2.594 & 0.326 & 4.178         & 6.240      & 3.652          & 4.304      \\
PEUGEOT                   & PAR      & 14.2  & 1.725 & 0.888 & 0.805 & 2.248 & 0.128 & 2.613         & 4.339     & 3.053          & 5.301      \\
ROYAL DUTCH SHELL  & AMS      & 83.0  & 0.677 & 0.529 & 0.350  & 0.753 & 0.137 & 1.206         & 1.883     & 1.103          & 1.856      \\
GBL                       & BRU      & 11.9  & 2.952 & 1.323 & 1.095 & 2.457 & 0.194 & 4.275         & 7.227     & 3.552          & 6.009      \\
KBC                       & BRU      & 18.5  & 2.124 & 0.702 & 0.716 & 2.881 & 0.217 & 2.826         & 3.944     & 3.597          & 4.018      \\
UCB                       & BRU      & 17.9  & 2.607 & 0.760  & 1.136 & 3.077 & 0.228 & 3.367         & 4.101     & 4.213          & 6.126      \\
VIVENDI                   & PAR      & 27.9  & 1.968 & 0.936 & 0.733 & 2.806 & 0.318 & 2.904         & 3.874     & 3.539          & 3.033     \\
\hline
Average & - & 39.7 &	2.27 & 1.02 & 	1.00 &	2.61 &	0.26 &	3.29 &	4.25 &	3.61 &	4.59

\end{tabular}
\caption{\label{table:results_large} Table of results, for the 50 stocks in our sample with the higher market cap. The column Exch. displays the exchange the stock is traded on (Amsterdam, Paris, Brussels or Lisbon) and the column Mcap shows the market capitalization of the stock in billions of euros (in October 2020). Then, $a_A$ and $a_B$ are the estimated tail exponents of sell and buy limit order distributions, for the left (l) and right (r) tail. $c$ is the estimator for the constant in equation (\ref{eq:delta}) and $\underline a=a_A + a_B$ without market orders (no MO), and $\underline a =\min(\frac{(c+1 )a_A}{c},a_A+2a_B)$ with market orders (MO), both displayed for left (l) and right (r) tails.}
\end{table}

\section{Conclusions}\label{sec:conclusion}
In this paper we study the tails of closing auction return distributions, both from a theoretical and empirical point of view, focusing on large closing price fluctuations.  Using the stochastic call auction model of \cite{Derksenetal20a}, we relate tail exponents of order placement distributions and tail exponents of the return distribution.  Empirical analysis supports the model's predictions. In theory, large market orders could be a cause of large closing price fluctuations, but this potential effect is cancelled by limit orders that are submitted against the direction of the market order imbalance.
Instead,  limit order placement appears to be the primary cause of observed heavy tails in closing auction return distributions.

\appendix
\section{Proofs}
\propcondtails*
\begin{proof}
The expression for the conditional distribution of $\underline X$ in equation (\ref{eq:clearingpricedist}), implies
\begin{align*}
\lim_{M \to \infty} &\frac{\PP(\underline{X}>M|N_A,N_B)}{T_B(M)T_A(M)^{N_A}} \\ 
& = \sum_{k=0}^{N_A} \sum_{l=k+1}^{N_B} \lim_{M \to \infty} {N_A \choose k} {N_B \choose l} \left(\frac{1-F_A(M)}{T_A(M)}\right)^{N_A} (1-F_A(M))^{-k} \\
&\quad \times \left(\frac{1-F_B(M)}{T_B(M)}\right) (1-F_B(M))^{l-1} \\
&= \sum_{k=0}^{N_A}  \sum_{l=k+1}^{N_B} \lim_{M \to \infty}{N_A \choose k}  {N_B \choose l} \left(\frac{1-F_B(M)}{1-F_A(M)}\right)^{k}  (1-F_B(M))^{l-1-k} \\
& = \sum_{k=0}^{N_A} \sum_{l=k+1}^{N_B}\lim_{M \to \infty}{N_A \choose k} {N_B \choose l} \left(\frac{T_B(M)}{T_A(M)}\right)^{k}  (1-F_B(M))^{l-1-k}\\
& = N_B,
\end{align*}
where we exchange limit and sum by dominated convergence, and the last line follows because all terms are 0, except when $l=1,k=0$.
\end{proof}
\thmtails*
\begin{proof}
The result of proposition \ref{prop:fa_heavier} implies 
\begin{align*}
\lim_{M \to \infty} \frac{\PP(\underline{X}>M)}{CT_A(M)T_B(M)} &= \lim_{M \to \infty} \frac{\mathbb{E} N_B T_B(M) T_A(M)^{N_A}}{CT_A(M)T_B(M)} \\
&= \lim_{M \to \infty} \frac{\sum_{i=1}^N \sum_{j=1}^N p_{N_A,N_B}(i,j) j T_B(M) T_A(M)^{i}}{CT_A(M)T_B(M)} \\
& = \sum_{i=1}^N \sum_{j=1}^N \frac{jp_{N_A,N_B}(i,j)}{C} \lim_{M \to \infty}T_A(M)^{i-1}\\
& = \frac{\sum_{j=1}^Njp_{N_A,N_B}(1,j)}{C}=1,
\end{align*}
where sum and limit are exchanged by dominated convergence and the last line follows by the fact that all terms in the sum are 0, except for $i=1$.
\end{proof}
\proptailsmos*
\begin{proof}
Suppose first that $\Delta-1>0$. Then
\begin{align*}
\lim_{M \to \infty} &\frac{\PP(\underline{X}>M|N_A,N_B,\Delta)}{T_A(M)^{N_A-(\Delta-1)}} \\ 
& = \sum_{k=0}^{N_A} \sum_{l=\max(k-\Delta+1,0)}^{N_B} \lim_{M \to \infty} {N_A \choose k} {N_B \choose l} \left(\frac{1-F_A(M)}{T_A(M)}\right)^{N_A-(\Delta-1)} \\
& \qquad \qquad \qquad \qquad \qquad \times (1-F_A(M))^{\Delta-1-k}  (1-F_B(M))^l \\
& = \sum_{k=0}^{N_A} \sum_{l=\max(k-\Delta+1,0)}^{N_B}\lim_{M \to \infty}{N_A \choose k} {N_B \choose l} \frac{T_B(M)^l}{T_A(M)^{k-\Delta+1}}\\
& = {N_A \choose \Delta-1},
\end{align*}
where the last line follows because all terms are 0, except when $k=\Delta-1,l=0$.\\
Now let $\Delta-1<0$, then
\begin{align*}
\lim_{M \to \infty} &\frac{\PP(\underline{X}>M|N_A,N_B,\Delta)}{T_A(M)^{N_A}T_B(M)^{1-\Delta}} \\ 
& = \sum_{k=0}^{N_A} \sum_{l=k+1-\Delta}^{N_B} \lim_{M \to \infty} {N_A \choose k} {N_B \choose l} \left(\frac{1-F_A(M)}{T_A(M)}\right)^{N_A-k} \\
 & \qquad \qquad \qquad \qquad \qquad \times T_A(M)^{-k} \left(\frac{1-F_B(M)}{T_B(M)}\right)^l T_B(M)^{l+\Delta-1} \\
& = \sum_{k=0}^{N_A} \sum_{l=k+1-\Delta}^{N_B}\lim_{M \to \infty}{N_A \choose k} {N_B \choose l} \frac{T_B(M)^{l+\Delta-1}}{T_A(M)^{k}}\\
& = {N_B \choose 1-\Delta},
\end{align*}
where the last line follows because all terms are 0, except when $k=0,l=1-\Delta$.
\end{proof}
\thmtailsmos*
\begin{proof}
Under assumptions \ref{ass:delta_general} and \ref{ass:fa_heavier_pl}, proposition \ref{prop:fa_heavier_plus_mo} transforms into,
\begin{align*}
&\PP(\underline{X}>M) \sim \sum_{n=1}^N \sum_{m=1}^N  \sum_{d=-N}^N K(n,m,d-1) M^{-(a_A(n-d+1) + (a_B-a_A)\max(-d+1,0))} p(n,m,d),&
\end{align*}
as $M \to \infty$. Here, we used that $\max(-x,0) -\max(x,0) = -x$, for all $x \in \RR$. By noting that $K(n,m,d)$ is bounded from above and below (by ${N \choose N/2}$ and 1), we see that the statement of the theorem holds true, for \\
$$\underline{a} = \min_{n,d:~ p(n,d)>0} (a_A (n-d+1) -(d-1)(a_B-a_A)\mathbf{1}_{\{d<1\}}),$$
where the minimum is taken over all $n,d$ such that $p(n,d) = \sum_m p(n,m,d)>0$.
Now note that the function $\underline{F}(n,d):= a_A (n-d+1) -(d-1)(a_B-a_A)\mathbf{1}_{\{d<1\}}$ is increasing in $n$, for every $d$. So the minimum is attained for the lowest $n$ with positive probability. Recall that we assumed $\Delta=c(N_A-N_B)$ and $N_B \in \{1,\dots,N\}$, so $p(n,d)=0$ for $n<\frac dc+1$, so the lowest $n$ with positive probability is $\hat{n}(d)=\max(\frac dc+1,1)$, for given $d$. Inserting into $\underline{F}$ leads to
$$\underline{F}(\hat{n}(d),d) = \begin{cases} a_A -(d-1) a_B &\text{ if } d\leq -1 \\
a_A((\frac 1c-1)d+2) & \text{ if }  d \geq 1 \end{cases},$$
which is minimal for $d=\pm1$, proving that
$$\overline{a} = \min\left(a_A\left(\frac 1c+1\right),a_A+2a_B\right)=\min\left(\frac{(c+1)a_A}{c}, a_A + 2a_B\right).$$
\end{proof}
\end{document}